\documentclass[journal]{ieeetran}
\usepackage{amsmath,amsthm,amssymb,paralist,subfigure,cite,graphicx,color,amsbsy,float}
\usepackage{pifont}
\usepackage{stmaryrd,mathrsfs,url}
\usepackage[table]{xcolor}
\usepackage{cuted,mathtools,lipsum}
\usepackage[ruled,vlined,linesnumbered]{algorithm2e}
\usepackage{soul}
\usepackage{microtype} 
\usepackage{graphicx}
\usepackage{hyperref}
\hypersetup{
	colorlinks=true,
	linkcolor=cyan,
	filecolor=mnodea,
	urlcolor=cyan,
	citecolor=lime,
}

\newtheorem{lem}{Lemma}
\newtheorem{thm}{Theorem}

\newtheorem{rem}{Remark}

\newtheorem{ass}{Assumption}

\def\mb{\mathbf}

\def\mc{\mathcal}

\DeclareMathOperator*{\argmin}{argmin}
\begin{document}
	\title{\huge \bf Log-Scale Quantization in Distributed First-Order Methods: Gradient-based Learning from Distributed Data
	}
	
	\author{Mohammadreza~Doostmohammadian,  Muhammad~I.~Qureshi, Mohammad Hossein Khalesi,\\ Hamid R. Rabiee,~\IEEEmembership{Senior Member,~IEEE}, Usman A. Khan,~\IEEEmembership{Senior Member,~IEEE} 
		
		\thanks{
			M. Doostmohammadian and M. H. Khalesi are with Faculty of Mechanical Engineering, Semnan University, Semnan, Iran, \texttt{\{doost,mhkhalesi\}@semnan.ac.ir}. 	
			M. I. Qureshi is with the Electrical and Computer Engineering Department, Tufts University, Medford, USA \texttt{muhammad.qureshi@tufts.edu}.
		   H. R. Rabiee is with the Computer Engineering Department, Sharif  University of Technology, Tehran, Iran \texttt{rabiee@sharif.edu}.
		   U. A. Khan holds concurrent appointments as a Professor of ECE/CS at Tufts University and as an Amazon Scholar. This paper describes work performed at Tufts University and is not associated with Amazon, \texttt{khan@ece.tufts.edu}.
	   }
	}
	\maketitle
\begin{abstract}
		Decentralized strategies are of interest for learning from large-scale data over networks. This paper studies learning over a network of geographically distributed nodes/agents subject to quantization. Each node possesses a private local cost function, collectively contributing to a global cost function, which the considered methodology aims to minimize. In contrast to many existing papers, the information exchange among nodes is log-quantized to address limited network-bandwidth in practical situations.
		We consider a first-order computationally efficient distributed optimization algorithm (with no extra inner consensus loop) that leverages node-level gradient correction based on local data and network-level gradient aggregation only over nearby nodes. This method only requires balanced networks with no need for stochastic weight design. It can handle log-scale quantized data exchange over possibly time-varying and switching network setups. We study convergence over both structured networks (for example, training over data-centers) and ad-hoc multi-agent networks (for example, training over dynamic robotic networks). 
		Through experimental validation, we show that (i) structured networks generally result in a smaller optimality gap, and (ii) log-scale quantization leads to a smaller optimality gap compared to uniform quantization. 
	\end{abstract}
\def\abstractname{Note to Practitioners}
\begin{abstract}
Motivated by recent developments in cloud computing, parallel processing, and the availability of low-cost CPUs and communication networks, this paper considers distributed and decentralized algorithms for machine learning and optimization. These algorithms are particularly relevant for decentralized data mining, where data sets are distributed across a network of computing nodes. A practical example of this is the classification of images over a networked data centre. In real-world scenarios, practical model nonlinearities such as data quantization must be addressed for information exchange among the computing nodes. This work emphasizes the importance of handling log-scale quantization and compares its performance over uniform quantization. By exploring these quantization methods, we aim to determine which is more accurate in terms of optimality gap and learning residual. Moreover, we study the impact of the structure of the information-sharing network on reducing the optimality gap and improving the convergence rate of distributed algorithms. As contemporary distributed and networked data mining systems demand highly accurate algorithms with fast convergence for real-time applications, our research emphasizes the benefit of structured networks under logarithmic quantization information-exchange. Our findings can be extended to different machine learning algorithms, offering pathways to more accurate and faster data mining solutions.	
\end{abstract}

	\section{Introduction}\label{sec_intro}
	\IEEEPARstart{L}{earning} 
	from geographically distributed data has emerged as a critical challenge in various domains such as sensor networks, edge computing, and decentralized systems \cite{raja2015cloud,abbasian2020survey,rahimian2023differentially}. The scenario involves individual nodes possessing private local cost functions, collectively contributing to a global cost that needs to be minimized \cite{yang2019survey,hashempour2021distributed,ddsvm,camponogara2010distributed}. This paper focuses on addressing this challenge through gradient-based optimization over networks, targeting scenarios where the cost functions are locally known at distributed nodes communicating over quantized channels.
	Data quantization, the process of representing data with a reduced number of bits, is often employed to reduce the communication load in such decentralized setups and is proven as an effective way to reduce communication costs, for example, in machine learning applications \cite{oh2021automated,sattler2021cfd}. One common quantisation type is logarithmic or log-scale that efficiently assigns more bits to represent smaller values and fewer bits to larger values \cite{kalantari2010logarithmic,cai2018deep}. This is beneficial when dealing with gradients or weights that often have a wide dynamic range. Smaller values, which might be more critical for convergence, are represented with higher precision. This may increase the complexity of the distributed learning algorithm as we move from linear dynamics to non-linear dynamics due to quantization. Uniform quantization, on the other hand, simplifies the representation of values, which can lead to a faster communication process. However, the lack of precision, especially for small values, might affect performance in scenarios where accurate representation of gradients or weights is crucial, and lead to residual error \cite{aysal2007distributed,pereira2008distributed}. Therefore, studying and comparing uniform and log-scale quantization for distributed optimization is of interest for distributed data mining applications.

	The existing literature on distributed optimization mostly consider linear function of the state variables, where the focus is mainly on reducing the optimality gap and improving the convergence rate. The primary distributed optimization methods are \textbf{GP} \cite{nedic2014distributed} and \textbf{DGD} algorithm \cite{hadjicostis2013average}, which are further extended to \textbf{DSGD} \cite{sundhar2010distributed} and \textbf{SGP} \cite{spiridonoff2020robust,nedic2016stochastic} by adopting stochastic gradient methods. Algorithm \textbf{ADDOPT} \cite{xi2017add} and its stochastic version \textbf{S-ADDOPT} \cite{qureshi2020s} further accelerate the convergence rate using gradient tracking. Similarly, momentum-based approaches improve the convergence rate using heavy-ball method \cite{nguyen2023geometric}, Nesterov gradient method \cite{nesterov2013introductory,qu2019accelerated}, or a combination of both \cite{nguyen2023accelerated}. On the other hand, \textbf{PushSVRG} \cite{9827792} and \textbf{PushSAGA} \cite{qureshi2021push,xin2021fast} algorithms use variance reduction to eliminate the uncertainty caused by the stochastic gradients. 
	Few other distributed optimization works use sign-based nonlinearity in the node dynamics to address finite-time and fixed-time convergence \cite{gong2021distributed,liu2022distributed2,chen2022distributed}.
	All these methods assume unlimited bandwidth to communicate arbitrary real numbers which in practice is not possible as all realistic communication protocols rely on finite precision.
	 %One key aspect missing in these works is the acknowledgement of quantized information exchange among nodes. This introduces a layer of complexity, in contrast to common approaches assuming linear communication patterns. 
	 The effect of uniform quantization over channels with limited data-rate and the associated optimality gap (or residual) is discussed in \cite{pu2016quantization,horvath2023stochastic,rabbat2005quantized}. The work \cite{xiong2022quantized} proposes a novel quantized distributed gradient tracking algorithm to minimize the finite sum of local costs over digraphs with linear convergence rate. The paper \cite{song2022compressed} combines the gradient tracking Push-Pull method with communication compression (quantized links). Similarly, communication compression technique compatible with a general class of operators unifying both unbiased and biased compressors is discussed in \cite{liao2022compressed}. Distributed quantized gradient-tracking under a dynamic encoding and decoding scheme is considered in \cite{kajiyama2020linear}.
	 What is missing from the existing literature is the consideration of log-scale quantization for distributed optimization to address resource-constrained communication networks with limited-bandwidth. This is primarily addressed for distributed resource allocation and constraint-coupled optimization \cite{scl}, but not for distributed learning over networks.
	 
	 %\ab{quantized optimization \cite{xiong2022quantized,song2022compressed,liao2022compressed,kajiyama2020linear} } 
	
	The main contributions of this paper are as follows: (i) we develop a log-scale quantized first-order distributed optimization method that employs network-level gradient tracking. This follows the real-world applications in which quantization is prevalent in networking setups. 
	%We further make a numerical comparison between uniform quantization versus logarithmic quantization in terms of optimality gap and convergence rate. 
	The proposed solution can be generalized to consider other sector-bound nonlinearities\footnote{A nonlinear function $f(x)$ is called sector-bound if it lies in a sector typically defined by two lines passing through the origin, creating an angular region. In mathematical form, there exist positive factors $\overline{\mc{K}}$ and $\underline{\mc{K}}$ for which we have $\underline{\mc{K}}x \leq f(x) \leq \overline{\mc{K}}x$.}  in the data-transmission channels. (ii) We develop the convergence analysis of the proposed log-scale quantized algorithm and provide the bound on the gradient-tracking step-size.    
	%Structured networks offer stability and predictability but may struggle to adapt to dynamic data distributions. Ad-hoc multi-agent networks, on the other hand, provide flexibility but may pose challenges in achieving cohesive global optimization due to their dynamic nature. An essential component of our investigation lies in understanding how optimization strategies can be tailored to suit these distinct network architectures. For structured networks, emphasis is placed on leveraging the predefined structure to optimize data exchange, enhance convergence rates, and reduce the optimality gap. In ad-hoc multi-agent networks, the challenge lies in dynamically adapting optimization strategies to the evolving network topology while ensuring global convergence.
	(iii) Moreover, unlike existing methods, our approach eliminates the need for \textit{weight-stochastic} design on the network. This is done by adopting a \textit{weight-balanced} design instead, simplifying implementation in the face of node/link failure and enhancing adaptability to \textit{time-varying} setups in the presence of link-level nonlinearities.  Through empirical validation, we analyze the convergence of our method for both academic and real-world data-learning setups with distributed and heterogeneous data while considering quantization. This work further explores optimization methodologies that can adapt to both structured networks and ad-hoc multi-agent networks. Structured networks embody a pre-defined organized framework, imposing order and predefined communication patterns on the participating nodes. In contrast, ad-hoc multi-agent networks dynamically adapt and emerge based on the specific demands of the distributed environment, for example, based on the nearest neighbour rule \cite{jadbabaie2003coordination}. This paper performs a comparative numerical analysis, unravelling the trade-offs and advantages inherent in using log-scale quantization in contrast to uniform quantization.
	
	The rest of the paper is organized as follows. Section~\ref{sec_prob} states the problem and terminologies. The proposed distributed algorithm and its convergence analysis are presented in Section~\ref{sec_alg}. Experimental validations are given in Section~\ref{sec_exp}. Finally, we conclude the paper in Section~\ref{sec_con}.
	
	\section{Assumptions, Terminology, and Problem Statement} \label{sec_prob}
	The distributed learning problem is defined as optimizing certain cost functions as a sum of some local cost functions distributed over a network of nodes/agents. In machine learning and data mining, the cost function (also known as the loss function or objective function) measures the difference between the predicted values of a model and the ground truth in the training dataset. The cost function quantifies how well or poorly a model performs by assigning a numerical value to the error or deviation between predicted and actual values. The optimization algorithm then adjusts the model's parameters to minimize this cost, effectively improving the model's performance. For example, different loss functions are used for data classification, regression, and support-vector-machine. The global cost is defined as the average of local losses at different nodes; mathematically, this is defined as
	\begin{align}
		\min_{\mb{x} \in \mathbb{R}^{p}} &
		F(\mb x) = \frac{1}{n}\sum_{i=1}^{n} f_i(\mb{x})\label{eq_prob}
	\end{align}   
	where each local loss $f_i(\mb{x})$ is private to node/agent $i$. Particularly, each local loss $f_i$ might be decomposable
	into a finite sum of $m_i$ component cost functions, i.e.,
	\begin{align}\label{eq_fij}
		f_i(\mb{x}) = \frac{1}{m_i}\sum_{j=1}^{m_i} f_{i,j}(\mb{x}).
	\end{align}
	In distributed learning and optimization setup, the training cost might be associated with error-back-propagation in Neural Networks \cite{qureshi2020s}, classification gap in SVM \cite{dsvm}, or regression training cost \cite{scaman2017optimal}.   
	
	We assume that the nodes share data (local gradient information or auxiliary variables) over a communication network $\mc{G}$. Each node locally communicates with its neighbouring nodes $j \in \mc{N}_i$ and shares relevant data over the links/channels. 
	In real-world practical applications, the information shared over the channels/links in $\mc{G}$ might be subject to quantization (or other nonlinear functions). The nonlinear mapping $h_l(\cdot)$ is  sign-preserving, odd, and sector-bound, i.e., it satisfies
	\begin{align}\label{eq_hl}
		0< \underline{\mc{K}} \leq \frac{h_l(z)}{z} \leq \overline{\mc{K}}
	\end{align}
	An example of such nonlinear mapping is logarithmic (or log-scale) quantization defined as
	\begin{align}\label{eq_hl_q}
		q_l(z) = \mbox{sgn}(z)\exp\left(\rho\left[\dfrac{\log(|z|)}{\rho}\right] \right),
	\end{align}
	with $[\cdot]$ as rounding to the nearest integer, and $\mbox{sgn}(\cdot)$, $\mbox{exp}(\cdot)$, and $\mbox{log}(\cdot)$ as the sign, exponential, and (natural) logarithm functions, respectively. The parameter $\rho$ is the logarithmic quantization level for which we have 
	\begin{align}\label{eq_hl_q2}
		\underline{\mc{K}} = 1-\dfrac{\rho}{2} \leq \dfrac{q_l(z)}{z} \leq 1+\dfrac{\rho}{2} = \overline{\mc{K}}.
	\end{align}	
	Note that, as illustrated in Fig.~\ref{fig_quant_desmos}, this nonlinearity is sector-bound, i.e., it is upper/lower-bounded by the two lines $(1 \pm \dfrac{\rho}{2})z$.
	This is in contrast to uniform quantization which is not sector-bound, as defined below 
	\begin{align}\label{eq_hl_qu}
		q_u(z) = \rho\left[\dfrac{z}{\rho}\right],
	\end{align}
	where $\rho$ is the uniform quantization level. The two quantization functions are represented in Fig.~\ref{fig_quant_desmos} for comparison. It can be observed that, in contrast to the log-scale quantization, the upper/lower bounds of uniform quantization, defined as $z \pm \dfrac{\rho}{2}$, are biased from the origin. In this work, we study the convergence under log-scale quantization as compared to uniform quantization. These two are the main types of the nonlinearities governing the data-transmission setup; for example see the following references on uniform quantization \cite{xiong2022quantized,song2022compressed,liao2022compressed,kajiyama2020linear} and on logarithmic quantization \cite{my_ecc,scl}. These nonlinearities are employed to reduce the communication bandwidth and network traffic in resource-constrained networking setups.
	\begin{figure} %[b]
		\centering
		\includegraphics[width=3.5in]{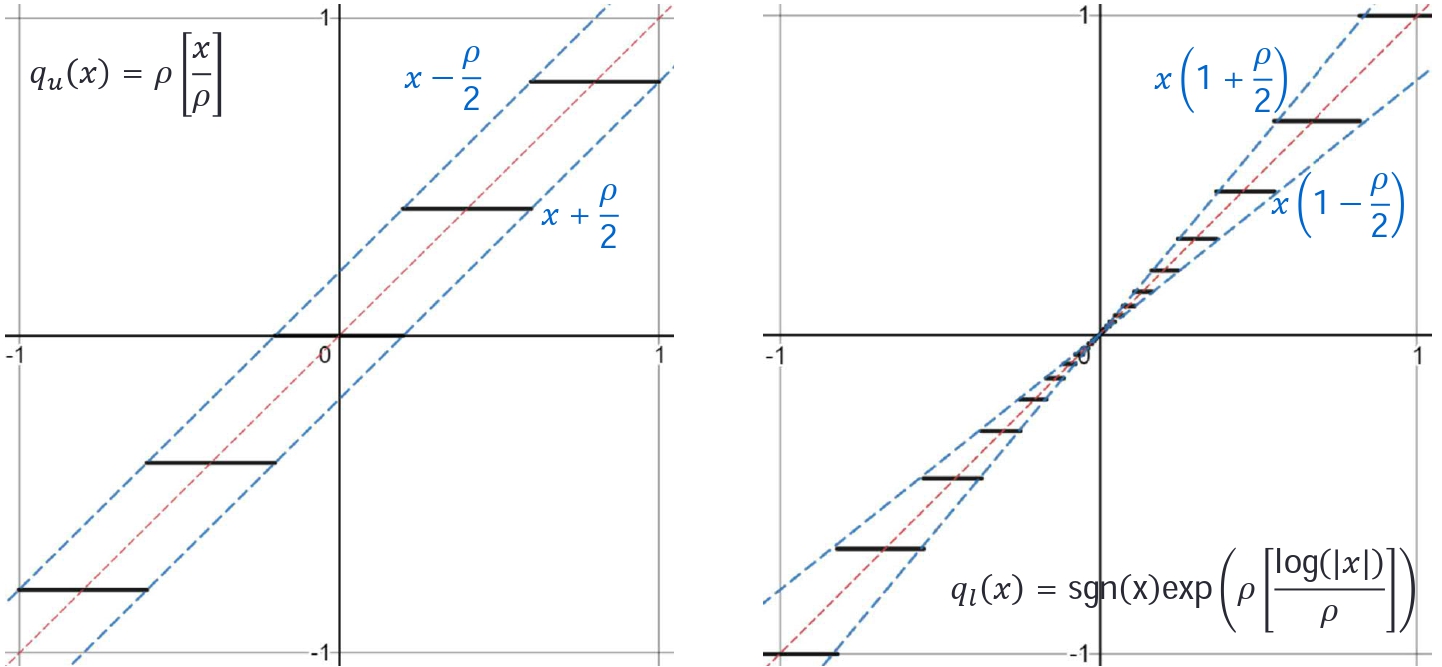}	
		\caption{Uniform quantization (as a non-sector-bound function) versus logarithmic quantization (as a sector-bound function). The logarithmic case leads to finer quantization around zero in contrast to the uniformly quantized case.  
		} \label{fig_quant_desmos}
	\end{figure}
	
	The assumptions to ensure convergence
	to the unique minimum $\mb{x}^*$ of the global cost function $F$ are as follows.
	\begin{ass} \label{ass_wb}
		The data-transmission network is connected (at all times), weight-balanced (WB), and could be time-varying.
	\end{ass}
    The WB condition implies that the sum of the weights on incoming links and outgoing links are equal at all nodes. This holds for every time-instant and for all the configurations of the network topology. To give a motivation behind this assumption, note that the WB condition is more relaxed than weight-stochasticity assumption in related works which require the link weights to sum to one for all nodes.
	\begin{ass}	\label{ass_conv}
		Each local cost function $f_i$ has $L$-smooth and is not necessarily convex, while the global cost $F$ is strongly convex.%, where $\kappa = L/\mu$ denotes its condition number.
	\end{ass}
	The above assumption is motivated by the discussion in \cite{Bottou} and is used in many references therein saying that it is common to use strong convexity of $F$ as many practical optimization objectives tend to be strongly convex near the minimum.
	These standard assumptions imply that the global cost function is differentiable and its second derivative is
	bounded above by some constant $L$ and below by strong-convexity constant $\mu$ where $0<\mu <L$. Further, the global optimizer $\mb{x}^*$
	of the cost $F$ exists and is also unique. The following subsection further motivates the study in this paper.
	%The time-varying property of the network further ensures convergence over volatile and dynamic network topologies, for example, in mobile multi-agent systems.

	\subsection{Effect of Network Structure and Quantization} \label{sec_effect}
	%Sample structured versus ad-hoc networks are given in Fig.~\ref{fig_graph}.
	%\begin{figure} %[b]
	%	\centering
	%	\includegraphics[width=2.5in]{fig_exp.jpg}
	%	\caption{(Left) a sample exponential graph of $n=16$ nodes (Right)
	%	} \label{fig_graph}
	%\end{figure}
	This paper investigates the effect of network structure by numerical experimentation. In distributed optimization and learning over networks, the choice of network topology plays a crucial role in determining the convergence rate and optimality gap of the learning algorithms \cite{nedic2018network,ying2021exponential}. Different network topologies can have varying effects on the efficiency and speed of convergence \cite{takezawa2024beyond,snam24,bandyopadhyay2018impact}. A general experimental comparison between structured networks (such as exponential graphs\footnote{In exponential graphs, as the number of nodes increases exponentially, the degree at each node increases linearly.}) and ad-hoc networking setups (such as geometric graphs\footnote{In geometric graphs, two nodes are connected if they are in physical proximity.}) is considered in this paper. Examples of structured networks include networked data centres or hierarchical social networks, and examples of ad-hoc random networks include collaborative robotic networks or wireless sensor networks.
	Structured networks (such as exponential graphs), characterized by well-defined and organized connections, can often lead to faster convergence rates \cite{ying2021exponential}. In exponential graphs, where nodes have a hierarchical structure, information can propagate more efficiently through the network \cite{lusher2013exponential}. Furthermore, structured networks may exhibit lower optimality gaps due to better coordination among the nodes \cite{nedic2018network}. The organized nature of the network facilitates synchronized updates, potentially leading to quicker convergence to a near-optimal solution.
	On the other hand, ad-hoc multi-agent networks (such as random geometric graphs), which might lack a clear structure, can have a more complex convergence behaviour. The convergence rate may be slower with a higher optimality gap compared to structured networks because information propagation might not follow an organized pattern. Nodes may require more iterations to align their models and achieve a consensus on the optimal solution \cite{tahbaz2006consensus}.
	%Note that, in general, the impact of network topology on convergence rate and optimality gap may vary based on the specific algorithm, problem, and the characteristics of the data.
	
	Another concern is to consider switching network topologies addressing 
	time-varying setups, for example, in mobile sensor networks and collaborative swarm robotics. The change in the network topology might be due to packet drops or link failures. Recall that most existing linear algorithms \cite{nedic2014distributed,hadjicostis2013average,sundhar2010distributed,spiridonoff2020robust,nedic2016stochastic,xi2017add,qureshi2020s,nguyen2023geometric,qu2019accelerated,nguyen2023accelerated,9827792,qureshi2021push}
	necessitate stochastic weight design which is prone to change in the network topology. This is because any change in the network structure violates the stochastic condition and requires the redesign of stochastic weights, for example by applying the algorithms in \cite{cons_drop_siam,6426252,datar2018memory}. In contrast, this work only requires WB design of matrices $A_\gamma$ and $B_\gamma$ that allows handling possible changes in the network topology. This is better illustrated in Fig.~\ref{fig_remov}. To preserve the WB condition in the presence of packet drops, we assume that if a packet is dropped on $i \rightarrow j$ link then $j \rightarrow i$ link is automatically deleted or, in other words, the agent $i$ does not apply the packet received from $j$ (for example, by setting the associated weight equal to zero), and the entire link is removed.
	This assumption is justified in existing literature \cite{6426252,olfatisaberfaxmurray07} saying that assuming common knowledge, both agents $i, j$ are aware
	of the delivery or loss of the packets over the
	mutual link. 
	\begin{figure}
		\centering
		\includegraphics[width=1.25in]{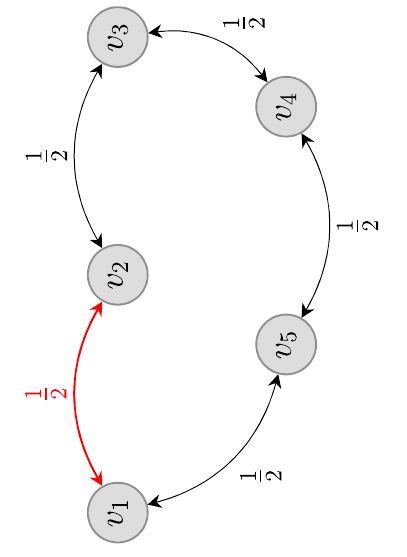}
		\includegraphics[width=1.25in]{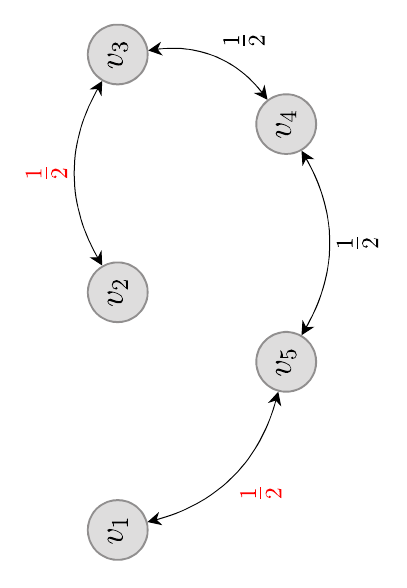}
		\caption{This figure gives an example graph topology which is both stochastic and WB. The red link represents an unreliable channel that might be subject to failure or packet drop. By removing this link, the network preserves the WB condition but loses stochasticity. Therefore, many weight-stochastic algorithms in the literature need to \textit{redesign} the weights for convergence. } \label{fig_remov}
	\end{figure}
	
	We address more findings in terms of data quantization over channels. First, it should be noted that logarithmic quantization, as compared to uniform quantization, is more efficient in scenarios where the optimization problem involves a wide range of gradient magnitudes. This follows the sector-bound nature of the log-scale quantization, as illustrated in Fig.~\ref{fig_quant_desmos}. Smaller values, which are key for convergence, can be represented more accurately (with more bits) under logarithmic quantization. This may lead to a smaller optimality gap and more efficient convergence, for example in optimal resource allocation \cite{ecc22}. The lack of precision in uniform quantization, especially for small values, might slow down convergence and increase the optimality gap in the range of values where accurate representation of gradients or weights is crucial \cite{oh2022non}. On the other hand, using finer bits for small values under log-quantization increases the communication complexity as compared to the uniform quantization \cite{xu2023aslog}. This implies that log-quantization, in general, adds more complexity to the communication network as compared to uniform quantization.
	In terms of network structure in the presence of quantization, although exponential graphs are structured and can facilitate efficient quantized communication, their optimality gap might be considerably affected by quantization inaccuracy. In other words, the efficacy of the network structure might be dominated by the effect of quantization. 
	These comments are investigated by the numerical experiments in Section~\ref{sec_exp}.

	\section{Distributed Optimization over Networks: Algorithm and Main Results} \label{sec_alg}
	\subsection{Algorithm}
	This section presents our main algorithm for learning over networks. The methodology is based on two WB matrices $A_\gamma=[a^\gamma_{ij}] \in \mathbb{R}^{n \times n}$ and $B_\gamma=[b^\gamma_{ij}] \in \mathbb{R}^{n\times n}$ that capture the weights for data-sharing over the network, one for consensus on the states and the other for gradient-tracking dynamics. The index $\gamma: k \mapsto \Gamma$ denotes the associated network topology that may change over time, whereas $\Gamma$ denotes the set of all possible topologies of the data-sharing network. The graph is all-time connected, i.e., all the network topologies in $\Gamma$ are connected.  The two matrices $A_\gamma$ and $B_\gamma$ can be designed independently, matrix $A_\gamma$ to characterize the weights on the local state value $\mb{x}_i$ and matrix $B_\gamma$ to characterize the weights on the auxiliary variable $\mb{y}_i$ for gradient tracking. These weight matrices are WB over $\mc{G}$ as discussed before,
	%, satisfying $\sum_{i=1}^{n} a^\gamma_{ij} = \sum_{j=1}^{n} a^\gamma_{ij}$ and $\sum_{i=1}^{n} b^\gamma_{ij} = \sum_{j=1}^{n} b^\gamma_{ij}$. 
	where the entries satisfy the following:
	\begin{align}
		A_\gamma(i,j) = \left\{
		\begin{array}{ll}
			a^\gamma_{ij}>0, & \text{If}~ j \in \mc{N}_i \\
			0, & \text{Otherwise}.
		\end{array}\right. \sum_{j=1}^{n} a_{ij}^\gamma=\sum_{i=1}^{n} a_{ij}^\gamma
	\end{align}
	\begin{align}
		B_\gamma(i,j) = \left\{
		\begin{array}{ll}
			b^\gamma_{ij}>0, & \text{If}~ j \in \mc{N}_i \\
			0, & \text{Otherwise}.
		\end{array}\right. \sum_{j=1}^{n} b_{ij}^\gamma=\sum_{i=1}^{n} b_{ij}^\gamma
	\end{align} 
	Let $\overline{A}_\gamma$ and $\overline{B}_\gamma$ denote the Laplacian matrices associated with $A_\gamma$ and $B_\gamma$, respectivly. The definition follows as
	\begin{align}
		\overline{A}_\gamma(i,j) = \left\{
		\begin{array}{ll}
			-\sum_{i=1}^{n} a^\gamma_{ij}, & i=j \\
			a^\gamma_{ij}, & i\neq j.
		\end{array}\right.
	\end{align}
	\begin{align}
		\overline{B}_\gamma(i,j) = \left\{
		\begin{array}{ll}
			-\sum_{i=1}^{n} b^\gamma_{ij}, & i=j \\
			b^\gamma_{ij}, & i\neq j.
		\end{array}\right.
	\end{align}
	It is known that for connected networks the eigenvalues of these Laplacian matrices are all in the left-half-plane except one zero eigenvalue \cite{olfatisaberfaxmurray07}.  The absolute value of the second largest eigenvalue of the laplacian matrices, denoted by $|\lambda^A_2|$ and $|\lambda^B_2|$, play a key role in the convergence and learning rate. These values are known as the algebraic connectivity.
	Given the weight matrices and switching signal $\gamma$, the local learning and optimization formulation is formally described below:
	\begin{align} \label{eq_xdot_g}	
		\dot{\mb{x}}_i &= \sum_{j=1}^{n} a_{ij}^\gamma (h_l(\mb{x}_j)-h_l(\mb{x}_i))-\alpha \mb{y}_i, \\ \label{eq_ydot_g}
		\dot{\mb{y}}_i &= \sum_{j=1}^{n} b_{ij}^\gamma (h_l(\mb{y}_j)-h_l(\mb{y}_i) ) + \partial_t \nabla f_i(\mb{x}_i),
	\end{align}
The solution by \eqref{eq_xdot_g}-\eqref{eq_ydot_g} is composed of two dynamics: one to derive the agents to reach agreement on the optimization variable $\mb{x}$ and one to track the consensus on the local gradients. For the latter, we introduce an auxiliary variable $\mb{y}$ that tracks the average of gradients in a consensus-based setup and this variable derive the first dynamics toward the global minimum, while both dynamics address log-quantized communication.  The log-scale quantization keeps the algorithm efficient in terms of communication with better performance as more resources are allocated near the origin.
%	and	its discretized version is as follows:
%	
%	\small
%	\begin{align} \label{eq_xdot_gd}	
%		{\mb{x}}^{k+1}_i &= {\mb{x}}^{k}_i +\sum_{j=1}^{n} \tilde{a}_{ij}^\gamma (h_l(\mb{x}_j^k)-h_l(\mb{x}_i^k))-\tilde{\alpha} \mb{y}^{k}_i, \\ \label{eq_ydot_gd}
%		{\mb{y}}^{k+1}_i &= {\mb{y}}^{k}_i +\sum_{j=1}^{n} \tilde{b}_{ij}^\gamma (h_l(\mb{y}^{k}_j)-h_l(\mb{y}^{k}_i) ) + \nabla f_i({\mb{x}}^{k+1}_i)-\nabla f_i({\mb{x}}^{k}_i),
%	\end{align} \normalsize
%	where $k$ denotes the iteration/epoch and $\alpha$ as the learning-rate is sufficiently small. In the discretized version, the elements $\tilde{a}_{ij}^\gamma = \eta {a}_{ij}^\gamma$, $\tilde{b}_{ij}^\gamma = \eta {b}_{ij}^\gamma$, and $\tilde{\alpha}=\eta {\alpha}$ are modified by the discrete step-time $\eta$.
In contrast to most existing methods with weight-stochastic matrices (introduced in Section~\ref{sec_intro}), the proposed dynamics work under WB design and under switching network topologies. These make it easier to handle link-failure as discussed in Fig.~\ref{fig_remov}. Moreover, the link (or channel) nonlinearity $h_l(\cdot)$ is not considered in the existing literature, which allows to address log-scale data quantization by setting $h_l(z)=q_l(z)$ given in Eq.~\eqref{eq_hl_q}.   
The discrete-time version of our methodology is summarized in Algorithm~\ref{alg_1}.
\begin{algorithm} \label{alg_1}
	\textbf{Data:}  $f_{i,j}(\mb{x})$, $\mc{G}_\gamma$, $A_\gamma$, $B_\gamma$, $\alpha$  \\	
	\textbf{Initialization:} ${\mb{y}}_i(0)=\mb{0}_{p}$, random ${\mb{x}}_i(0)$
	\\
	\For{$k=0,1,2,\dots$}{
		Each node $i$ receives $h_l(\mb{x}_j^k)$ and $h_l(\mb{y}^{k}_j)$ from $j \in \mc{N}_i$\;
		$\mb{x}^{k+1}_i \leftarrow {\mb{x}}^{k}_i +\sum_{j=1}^{n} a_{ij}^\gamma (h_l(\mb{x}_j^k)-h_l(\mb{x}_i^k))-\alpha \mb{y}^{k}_i$\;
		$w_i^k \leftarrow \nabla f_i({\mb{x}}^{k+1}_i)-\nabla f_i({\mb{x}}^{k}_i)$\;
		$\mb{y}^{k+1}_i \leftarrow {\mb{y}}^{k}_i +\sum_{j=1}^{n} b_{ij}^\gamma (h_l(\mb{y}^{k}_j)-h_l(\mb{y}^{k}_i) ) + w_i^k$\;
		%$h_l(\mb{x}_i^{k+1})$ and $h_l(\mb{y}^{k+1}_i)$ are shared over $\mc{G}_\gamma$\;
	}
	\textbf{Return:}  optimal $\mb{x}^*$ and $F^*$\;	
	\caption{Local learning at each node $i$. }
\end{algorithm} 

	Note that the learning rate tightly depends on the network topology and the algebraic connectivity. 
	%Note that, the proof of convergence follows similarly to the linear case in \cite{ddsvm} without any non-ideal linking conditions,  and it is skipped here due to space limitation. 
	Summing the states over all the nodes we have 
	\begin{align}  \label{eq_sumydot}
		\sum_{i=1}^n {\mb{y}}_i^{k+1} 
		&= \sum_{i=1}^n {\mb{y}}_i^{k} + \sum_{i=1}^n \nabla f_i(\mb{x}^{k+1}_i)-\sum_{i=1}^n \nabla f_i(\mb{x}^{k}_i), \\
		\label{eq_sumxdot}
		\sum_{i=1}^n {\mb{x}}_i^{k+1} 
		&= \sum_{i=1}^n {\mb{x}}_i^{k}-\alpha \sum_{i=1}^n {\mb{y}}_i^{k}.
	\end{align}
	This follows the WB assumption which implies that 
	$$\sum_{i=1}^{n} \sum_{j=1}^{n} a_{ij}^\gamma (h_l(\mb{x}_j^k)-h_l(\mb{x}_i^k))=0$$
	$$\sum_{i=1}^{n} \sum_{j=1}^{n} b_{ij}^\gamma (h_l(\mb{y}^{k}_j)-h_l(\mb{y}^{k}_i))=0$$
	By setting the initial condition as~$\mb{y}_i(0)=\mb{0}_{p}$ we obtain the following, 
	\begin{eqnarray} \label{eq_sumxdot2}
		\sum_{i=1}^n {\mb{x}}_i^{k+1} = -\alpha \sum_{i=1}^n {\mb{y}}_i^{k+1} = -\alpha \sum_{i=1}^n \nabla f_i(\mb{x}^{k+1}_i),
	\end{eqnarray}
	This represents the consensus-type gradient-tracking dynamics that capture the essence of the proposed methodology. This formulation further allows for addressing sector-bound odd nonlinearities $h_l$ (for example log-scale quantization) without violating the tracking property. This along with the time-varying network topology advances the recent state-of-the-art algorithms for data-mining and optimization \cite{qureshi2020s,9827792,qureshi2021push,pu2016quantization,horvath2023stochastic,rabbat2005quantized}.

	Recall that most of the existing algorithms prove convergence for ideal networking conditions and in the absence of nonlinearities (e.g., quantization and clipping). However, in the presence of nonlinearities and non-ideal networking conditions, the solution may result in a certain optimality gap, which also depends on the structure of the underlying network $\mc{G}_\gamma$. The role of the network structure and quantization is also discussed in the simulation section.
	
	\subsection{Proof of Convergence} \label{sec_proof}
	This section proves the convergence of the distributed dynamics~\eqref{eq_xdot_g}-\eqref{eq_ydot_g} using eigen-spectrum perturbation-based analysis. We also state some relevant lemmas. The proofs in this section hold for any sign-preserving, odd, sector-bound nonlinear mapping $h_l(\cdot)$, which also includes log-scale quantization $h_l(z)=q_l(z)$ defined by \eqref{eq_hl_q}.  
	In the rest of the paper, for notation simplicity, we drop the dependence $(t,\gamma)$ unless where it is needed. 
	
	\begin{lem} \label{lem_dM} ~\cite{cai2012average} Consider the square matrix~$P(\alpha)$ of size $n$ which depends on parameter~${\alpha \in \mathbb{R}_{\geq0} }$. Let~$P(0)$ has~${N<n}$ equal eigenvalues~$\lambda_1=\ldots=\lambda_N$, associated with right and left unit eigenvectors~$\mb{v}_1,\ldots,\mb{v}_N$ and~$\mb{u}_1,\ldots,\mb{u}_N$ (which are linearly independent). Let~${P' = \partial_{\alpha} P(\alpha)|_{\alpha=0}}$ and $\lambda_i(\alpha)$ represent its $i$-th eigenvalue~${\lambda_i, i \in \{1,\ldots,N\}}$. Then,~$ \partial_{\alpha}\lambda_{i}|_{\alpha=0}$ is the $i$-th eigenvalue of, % the following~$l$-by-$l$ matrix,
		\[\left(\begin{array}{ccc}
			\mb{u}_1^\top P' \mb{v}_1 & \ldots & \mb{u}_1^\top P' \mb{v}_N \\
			& \ddots & \\
			\mb{u}_N^\top P' \mb{v}_1 & \ldots & \mb{u}_N^\top P' \mb{v}_N
		\end{array} \right).
		\]
	\end{lem}
	
	\begin{lem} \label{lem_Wg2}
		\cite{SensNets:Olfati04} Given the balanced digraph $\mc{G}$ its Laplacian matrix $\overline{A}$ (or $\overline{B}$) has all its eigenvalue in the left-half-plane except one isolated zero eigenvalue with associated non-negative left eigenvector $\mb{u}_1^\top$ satisfying $\sum_{i=1}^n u_{1,i} >0$.
	\end{lem}
	
	\begin{thm} \label{thm_zeroeig}
		Let Assumptions~\ref{ass_wb} and \ref{ass_conv} hold. For sufficiently small~$\alpha$, all eigenvalues of the dynamics~\eqref{eq_xdot_g}-\eqref{eq_ydot_g} are in the left-half-plane except $m$ zero eigenvalue.
	\end{thm}
	\begin{proof}
		For the linear case (i.e., setting $h_l(z)=z$) the proposed dynamics~\eqref{eq_xdot_g}-\eqref{eq_ydot_g} can be written in compact form as		
		\begin{align} \label{eq_xydot}
			&\left(\begin{array}{c} \dot{\mb{x}} \\ \dot{\mb{y}} \end{array} \right) = M(\alpha) \left(\begin{array}{c} {\mb{x}} \\ {\mb{y}} \end{array} \right),
			\\ \label{eq_M}
			M(\alpha ) = &\left(\begin{array}{cc} \overline{A} \otimes I_m & -\alpha I_{mn} \\ H(\overline{A}\otimes I_m) & \overline{B} \otimes I_m - \alpha H
			\end{array} \right).
		\end{align}
		where~$H:=\mbox{diag}[\nabla^2 f_i(\mb{x})]$ and $M(\alpha)=M_0+\alpha M_1$ with
		\begin{eqnarray}\nonumber
			M_0 &=&  \left(\begin{array}{cc} \overline{A} \otimes I_m & \mb{0}_{mn\times mn} \\ H(\overline{A} \otimes I_m) & \overline{B}\otimes I_m \end{array} \right),\\\nonumber
			M_1 &=& \left(\begin{array}{cc} \mb{0}_{mn\times mn} & - {I_{mn}} \\ {\mb{0}_{mn\times mn}} & - H \end{array} \right),
		\end{eqnarray}
		Now, considering the nonlinearity $h_l(z) = q_l(z)$\footnote{Note that the proof holds for general sign-preserving sector-bound nonlinearities. We particularly state the solution for log-scale quantization as an example.}, one can linearize the nonlinear dynamics~\eqref{eq_xdot_g}-\eqref{eq_ydot_g} at every time-instant $t$ as its operating point. 
		It is known that the stability of the linearization at every operating point implies the stability of the nonlinear dynamics \cite{nonlin}. To study the stability of the linearized dynamics, we have  
		\begin{align}  \label{eq_Mg}
			M_q(t,\alpha,\gamma) &=   M_q^0 + \alpha M^1 \\ \label{eq_beta_M0}
			\underline{\mc{K}} M^0 & \preceq M_q^0 \preceq \overline{\mc{K}} M^0 \\ \label{eq_beta_M}
			M_q^0 = Z(t)  M^0 &,~ \underline{\mc{K}} I_n  \preceq Z(t) \preceq \overline{\mc{K}} I_n 
		\end{align}
		where $M_q^0$ is the linearized version of $M^0$ for the nonlinear dynamics, $Z(t) := \mbox{diag}[\zeta(t)]$, column vector $\zeta(t) = [\zeta_1(t);\zeta_2(t);\dots;\zeta_n(t)]$ with $\zeta_i(t) = \frac{q_l(\mb{x}_i)}{\mb{x}_i}$ (or equivalently $q_l(\mb{x}(t)) = Z(t) \mb{x}(t)$ at every time-instant $t$). Recall from Eq.~\eqref{eq_hl} that  $\underline{\mc{K}} \leq \zeta_i(t) \leq \overline{\mc{K}}$. Therefore, we have
		\begin{align}  
			\overline{A}_{Z} =  \overline{A} Z(t),~\overline{B}_{Z} =  \overline{B} Z(t) 
		\end{align}
		The above helps to relate the eigen-spectrum of the linear dynamics~\eqref{eq_xydot} to that of the nonlinear case associated with $M_q$. Note that $Z(t)$ is a diagonal matrix from its definition, and we have
		\begin{align}  
			\mbox{det}(\overline{A}_{Z}-\lambda I_{mn}) = \mbox{det}(\overline{A}-\lambda Z(t)^{-1}), \\
			\mbox{det}(\overline{B}_{Z}-\lambda I_{mn}) = \mbox{det}(\overline{B}-\lambda Z(t)^{-1}),
		\end{align}
		with $\lambda$ denoting the eigenvalue. Therefore, from \eqref{eq_beta_M0}-\eqref{eq_beta_M}, one can relate the eigenspectrum of the linearized version with that of the linear case as
		\begin{align} \label{eq_spect_k}
			\underline{\mc{K}} \sigma(M^0) \leq \sigma(M^0_q) \leq \overline{\mc{K}} \sigma(M^0)
		\end{align}
		where the eigenspectrum of the linear case follows the block triangular form of~$M_0$ as,
		\begin{align} \label{eq_sigma}
			\sigma(M^0) = \sigma(\overline{A} \otimes I_m) \cup \sigma(\overline{B} \otimes I_m)
		\end{align}
		For the rest of the proof, we study the eigenspectrum perturbation of~$M^0$ and then extend the results to~$M^0_q$ using Eq.~\eqref{eq_spect_k}.	      
		From Lemma~\ref{lem_Wg2}, both Laplacian matrices ~$\overline{A}$ and~$\overline{B}$ have one isolated zero eigenvalue and the rest in the left-half-plane for all cases of the switching networks. Therefore, $m$ sets of (possibly changing) eigenvalues of~$M_0$, associated with dimensions~$j=\{1,\ldots,m\}$ are in the form,
		$$\operatorname{Re}\{\lambda_{2n,j}\} \leq \ldots \leq \operatorname{Re}\{\lambda_{3,j}\} < \lambda_{2,j} = \lambda_{1,j} = 0,$$
		Next, we consider~$\alpha M_1$ as a perturbation to matrix $M_0$. This is done using Lemma~\ref{lem_dM}. In particular, we check how the perturbation term~$\alpha M_1$ affects the two zero eigenvalues~$\lambda_{1,j}$ and~$\lambda_{2,j}$ of $M_0$. Denote the perturbed zero eigenvalues  by~$\lambda_{1,j}(\alpha)$ and~$\lambda_{2,j}(\alpha)$ associated to $M$. For all WB and connected switching network topologies, the right and left unit eigenvectors of~$\lambda_{1,j}$,~$\lambda_{2,j}$ follow from Lemma~\ref{lem_Wg2} and \cite{SensNets:Olfati04} as\footnote{The normalizing factors of the unit vectors might be ignored as in the followings we only care about the sign of the terms in $U^\top M_1 V$, not the exact values. },
		\begin{align} \label{eq_V}
			% V &= [V_1~ V_2] =\left(\begin{array}{cc}
			% 	\frac{1}{n} \mb{1}_n& \mb{0}_n \\
			% 	\mb{0}_n & \mb{v}_2
			% \end{array} \right)\otimes I_m
			V &= [V_1~ V_2] =
			\frac{1}{\sqrt{n}} \left(\begin{array}{cc}
				\mb{1}_n& \mb{0}_n \\
				\mb{0}_n & \mb{1}_n
			\end{array} \right) \otimes I_m
			\\ \label{eq_U}
			U^\top &= [U_1~ U_2]^\top =\left(\begin{array}{cc}
				\mb{u}_1& \mb{0}_n \\
				\mb{0}_n & \mb{u}_2
			\end{array} \right)^\top \otimes I_m
		\end{align}
		Recalling $M(\alpha)=M_0+\alpha M_1$ we have~$\partial_{\alpha}{dM(\alpha)}|_{\alpha=0}=M_1$. Then, Lemma~\ref{lem_dM} implies that,
		\begin{align} \label{eq_dmalpha}
			U^\top M_1 V= \left(\begin{array}{cc}
				\mb{0}_{m\times m}	& \times  \\
				\mb{0}_{m\times m}	& -(\mb{u}_2 \otimes I_m)^\top H  (\frac{1}{\sqrt{n}}\mb{1}_n \otimes I_m)
			\end{array} \right).
		\end{align}
		The definition of~$H$, Lemma~\ref{lem_Wg2}, and Assumption~\ref{ass_conv} implies that 
		\begin{align} \label{eq_sum_df}
			-(\mb{u}_2 \otimes I_m)^\top H  (\frac{1}{\sqrt{n}}\mb{1}_n \otimes I_m) = -\sum_{i=1}^n  \frac{u_{2,i}}{\sqrt{n}} \boldsymbol{  \nabla}^2 f_i(\mb{x}_i) \prec 0,
		\end{align}
		where $u_{2,i}$ is the $i$th element of $\mb{u}_2$.
		Then, using Lemma~\ref{lem_dM}, the perturbations ~$\partial_{\alpha} d\lambda_{1,j}|_{\alpha=0}$ and~$\partial_{\alpha} d\lambda_{2,j}|_{\alpha=0}$ are defined based on the eigenvalues of the  matrix in~\eqref{eq_dmalpha}. Since this matrix is triangular and from Eq.~\eqref{eq_sum_df},~$\partial_{\alpha}\lambda_{1,j}|_{\alpha=0} = 0$ and~$\partial_{\alpha}\lambda_{2,j}|_{\alpha=0} < 0$. Therefore, for all sets of switching network Laplacians, the perturbation analysis imply that~$\alpha M_1$ pushes $m$ zero eigenvalues~$\lambda_{2,j}(\alpha)$ of~$M$ toward the left-half-plane and the other $m$ zero eigenvalues~$\lambda_{1,j}(\alpha)$ remain at zero. This implies that sufficiently small~$\alpha$ gives the eigen-spectrum of $M$ as
		
		\small \begin{align}
			\begin{aligned}
				\operatorname{Re}\{\lambda_{2n,j}(\alpha)\} \leq \ldots \leq \operatorname{Re}\{\lambda_{3,j}(\alpha)\}
				\leq \lambda_{2,j}(\alpha) < \lambda_{1,j}(\alpha) = 0,
			\end{aligned}
		\end{align} \normalsize
		This completes the proof.
	\end{proof}

	\begin{thm} \label{thm_lyapunov}
		Given that the Assumptions \eqref{ass_wb}-\eqref{ass_conv} and conditions in Theorem~\ref{thm_zeroeig} hold, for $0 < \alpha <  \frac{\min \{|\operatorname{Re}\{\lambda_2^A\}|,| \operatorname{Re}\{\lambda_2^B\}|\}}{L\overline{\mc{K}}}$, the proposed dynamics~\eqref{eq_xdot_g}-\eqref{eq_ydot_g} converges to the optimizer~$[\mb{x}^*\otimes \mb{1}_n;\mb{0}_{nm}]$.
	\end{thm}
	\begin{proof}
			First, we determine the admissible bound on $\alpha$ for proof of convergence.
		Recall from \cite[Appendix]{delay_est} that one can relate the spectrum of $M(\alpha )$ in \eqref{eq_M} to $\alpha$. For ease of notation, from this point onward we derive the bound for $m=1$ (but it holds for any $m>1$). Performing row/column permutations in \cite[Eq.~(18)]{delay_est}  $\sigma(M)$ can be determined from the following, % (RECHECK):
		\begin{align} \nonumber
			\mbox{det}(\alpha  I_{n}) \mbox{det}(H(\overline{A}) +(\overline{B}  - \alpha H -\lambda I_{n}) (\frac{1}{\alpha})(\overline{A}  -\lambda I_{n})) = 0.
		\end{align} 
		which can be simplified as
		\begin{align} \label{eq_m=1}
			\mbox{det}(I_{n}) \mbox{det}((\overline{A}  -\lambda I_{n})(\overline{B}  - \lambda I_{n}) +\alpha \lambda H ) = 0
		\end{align}
		Similar to the proof of Theorem~\ref{thm_zeroeig}, for stability we need to find the admissible range of  $\alpha$ values for which the eigenvalues $\lambda$ remain in the left-half-plane, except one zero eigenvalue. The analysis here is based on the fact that the eigenvalues are continuous functions of the matrix elements~\cite{stewart_book}.
		% Using the above remark and some simplifications, \eqref{eq_m=1}  changes to,
		% \begin{align} \label{eq_m=1_sym}
		% \mbox{det}((\overline{A}^{sym}_{q}  -\lambda I_{n})(\overline{W}^{sym}_{q}  - \lambda I_{n}) +\alpha \lambda H ) = 0
		% \end{align}
		% with real eigenvalues $\lambda$.
		It is clear  that $\alpha=0$ satisfies Eq. \eqref{eq_m=1} and gives $\mbox{det}((\overline{W}  - \lambda I_{n})(\overline{A}  - \lambda I_{n})) = 0$. This leads to the eigen-spectrum following as $\sigma(M) = \sigma(\overline{A}) \cup \sigma(\overline{W})$ with two zero eigenvalues for all switching network topologies. We need to find the other root $\overline{\alpha}>0$, which gives the admissible range as $0<\alpha<\overline{\alpha}$ for the stability of $M(\alpha)$. This follows the continuity of $\sigma(M)$ as a function of $\alpha$ \cite{stewart_book}.
		With some abuse of notation, for any $\operatorname{Re}\{\lambda\}<0$ one can reformulate \eqref{eq_m=1} as
		\begin{align} \nonumber
			\mbox{det}(&(\overline{B}  -\lambda I_{n} \pm \sqrt{\alpha |\lambda| H})(\overline{A}  - \lambda I_{n} \mp \sqrt{\alpha |\lambda| H} ) \\ &\pm \sqrt{\alpha |\lambda| H}(\overline{A}  -\overline{B})) = 0 \label{eq_det_all}
		\end{align}
		This follows the diagonal form of $H$. First, we consider the same consensus matrix for both variables $\mb{x},\mb{y}$.
		Set $\overline{A} = \overline{B}$ and we have
		\begin{align} \nonumber
			\mbox{det}(\overline{A}  - \lambda I_{n} \mp \sqrt{\alpha |\lambda| H} )= \mbox{det}(\overline{A}  -\lambda (I_{n}   \mp \sqrt{\frac{\alpha  H}{|\lambda|}} )) = 0
		\end{align}
		This implies that fo $\lambda \in \sigma(\overline{A})$, its perturbed eigenvalue is $\lambda (1 \pm \sqrt{\frac{\alpha  H}{|\lambda|}})$. Therefore, for $\lambda \neq 0$, the min value of $\overline{\alpha}$ that makes this term zero (at the edge of instability) satisfies the following,
		\begin{align}
			\overline{\alpha} = \argmin_{\alpha} |1 - \sqrt{\frac{\alpha  H}{|\lambda|}}|
			%    = \frac{\min \{|\lambda_j|\neq 0\}}{\max \{H_{ii}\}} = \frac{|\lambda_2|}{\gamma}
			\geq \frac{\min \{|\operatorname{Re}\{\lambda_j\}|\neq 0\}}{\max \{H_{ii}\}} = \frac{|\operatorname{Re}\{\lambda_2\}|}{L}
		\end{align}
		where we recalled $H \preceq L I_{n}$ from Assumption~\ref{ass_conv}.
		%This implies that the (real part of) eigenvalues are perturbed towards the RHP by $\sqrt{\alpha |\lambda| H}$ and one can find the other root value $\overline{\alpha}$ from $\lambda I_{n} = \sqrt{\alpha |\lambda| H}$. This gives $\overline{\alpha} \geq \frac{\min \{|\operatorname{Re}\{\lambda_j\}|\neq 0\}}{\max \{H_{ii}\}} = \frac{|\operatorname{Re}\{\lambda_2\}|}{\gamma}$ for $H \preceq \gamma I_{n}$.
		Recalling perturbed formulation in Theorem~\ref{thm_zeroeig}, this gives the admissible range for $\alpha$ for which $M(\alpha)$ has all its eigenvalues at the left-half-plane except one zero as
		\begin{align} \label{eq_alphabar0}
			0 < \alpha < \overline{\alpha}:= \frac{|\operatorname{Re}\{\lambda_2\}|}{L}
		\end{align}
		For general nonlinear dynamics~\eqref{eq_xdot_g}-\eqref{eq_ydot_g} satisfying Eq.~\eqref{eq_beta_M0}-\eqref{eq_beta_M}, the admissible range changes to
		\begin{align} \label{eq_alphabar00}
			0 < \alpha < \overline{\alpha}:= \frac{|\operatorname{Re}\{\lambda_2\}|}{L\overline{\mc{K}}}
		\end{align}
		For the more general case of $\overline{A} \neq \overline{B}$, this is generalized as,
		\begin{align} \label{eq_alphabar}
			0 < \alpha <  \frac{\min \{|\operatorname{Re}\{\lambda_2^A\}|,| \operatorname{Re}\{\lambda_2^B\}|\}}{L\overline{\mc{K}}} =: \overline{\alpha}
		\end{align}
		This admissible range of $\alpha$ depends only on $\sigma(\overline{A}),\sigma(\overline{B})$ and holds for any value of $m \geq 1$.
		Next, we prove convergence under this admissible bound. It is clear that the optimizer~$[\mb{x}^*\otimes \mb{1}_n;\mb{0}_{nm}]  \in \mathbb{R}^{2mn}$is the invariant set of dynamics~\eqref{eq_xdot_g}-\eqref{eq_ydot_g} and belongs to null-space of $M_q$. Recall that for $0 < \alpha <  \frac{\min \{|\operatorname{Re}\{\lambda_2^A\}|,| \operatorname{Re}\{\lambda_2^B\}|\}}{L\overline{\mc{K}}}$, the eigenvalues of the linearized system dynamics associated with $M_q$ are stable except one isolated zero eigenvalue. This is irrespective of the time-variation of $\mc{G}$ and $M_q$. Define the variable $\delta$ as the distance between the system state and the optimizer,
		$$\delta  = \left(\begin{array}{c} {\mb{x}} \\ {\mb{y}} \end{array} \right) - \left(\begin{array}{c} {\mb{x}}^* \\ \mb{0}_{mn} \end{array} \right) \in \mathbb{R}^{2mn}.$$
		and the positive-semi-definite Lyapunov function ${\mc{V}(\delta) = \frac{1}{2} \delta^\top \delta =  \frac{1}{2}\lVert \delta \rVert_2^2}$. As $\delta \rightarrow \mb{0}_{2mn}$ the state variables converge to the optimizer. We have, 
		\begin{align} \label{eq_mdelta}
			\dot{\delta} = M_q \left(\begin{array}{c} \mb{x} \\ \mb{y} \end{array} \right) - M_q\left(\begin{array}{c} \mb{x}^* \\ \mb{0}_{mn} \end{array} \right) = M_q \delta.
		\end{align}
		This implies that $\dot{\mc{V}} = {\delta}^\top \dot{\delta}=  \delta^\top M_q {\delta}$ and from~\cite[Sections~VIII-IX]{SensNets:Olfati04}, we have 
		\begin{eqnarray} \label{eq_Re2}
			\dot{\mc{V}} = \delta^\top M_q \delta \leq \max_{1\leq j\leq m}\operatorname{Re}\{{\lambda}_{2,j}\} \delta^\top  \delta, 
		\end{eqnarray}
		where $\max_{1\leq j\leq m}\operatorname{Re}\{{\lambda}_{2,j}\}$ is the real part of the largest  negative eigenvalue of $M_q$ at all times (following from Theorem~\ref{thm_zeroeig}).
		Note that $\dot{\mc{V}}$ is negative-definite for all $\delta \neq \mb{0}_{2mn}$ and is zero for $\delta = \mb{0}_{2mn}$. From Lyapunov theorem, this network of integrators is asymptotically globally stable and $\delta \rightarrow \mb{0}_{2mn}$ implies that the solution converges to the optimizer. 
	\end{proof}
	%For undirected networks, since the eigenvalues of $M$ are real, one extends the proof to hybrid setups under a switching signal $q$.
	
	\subsection{Discussions}
	We now provide some important remarks in the context of Theorems~\ref{thm_zeroeig} and \ref{thm_lyapunov} and the sector-bound property.
	\begin{rem}
	For convergence we only need $(\mb{1}_n \otimes I_p)^\top H (\mb{1}_n \otimes I_p) \succ 0$, i.e., strong-convexity of the global objective function $F(\mb x)$. Thus, the local cost functions $f_i(\mb{x}_i)$ might be non-convex, i.e., $\nabla^2 f_i(\mb x_i)$ is not necessarily positive while their summation is positive to ensure that $F(\mb x)$ is strongly convex. An academic example of such a case is given in the simulation. Therefore, the convexity assumption in this paper is more relaxed as compared to many existing literature. Note that the log-scale quantization in this work is independent of Assumption~\ref{ass_conv}.  
\end{rem}	
	\begin{rem}
	From Theorem~\ref{thm_zeroeig}, the proposed dynamics~\eqref{eq_xdot_g}-\eqref{eq_ydot_g} has one zero eigenvalue related to the agreement (or consensus) on the states, i.e., all the auxiliary states reach consensus on $\mb{y}_i = 0$ and all the main states reach consensus on $\mb{x}_i = \mb{x}^*$. On the other hand, the rest of the eigenvalues remain in the left-half-plane which results in the stability of the proposed dynamics and convergence toward the optimal point.  
    \end{rem}  
	
	\begin{rem}\label{rem_rate}
		Recall that one can approximate the convergence rate (or decay rate) of the linearized system dynamics given by Eq.~\eqref{eq_mdelta} which is concluded from the dynamics~\eqref{eq_xdot_g}-\eqref{eq_ydot_g}. Following the eigen-spectrum of the system matrix $M_q(\alpha)$ and using Eq.~\eqref{eq_Re2}, the convergence rate is (at least) $\exp(\max_{1\leq j\leq m}\operatorname{Re}\{{\lambda}_{2,j}\} t)$. This shows linear convergence in log-scale which is approved by the simulations in the next section.
		%	Consider the typical example of link nonlinearity $h_l(\cdot)$ as log-scale quantization on the communication channels. Following Eq.~\eqref{eq_hl_q2}, for sufficiently small $\rho$ one can see that $q_l(\cdot)$ gets arbitrarily close to the linear case.  The 
		Also, for the special case of no quantization (i.e., $h_l(\mb{x}_i)=\mb{x}_i$ and $h_l(\mb{y}_i)=\mb{y}_i$) and for strongly convex cost functions with \textit{stochastic} weight matrices $A_\gamma,B_\gamma$ the convergence rate can be calculated via similar analysis as in  \cite[Lemmas~ 3.2-3.4]{jakovetic2018convergence}.  
	\end{rem}

	\begin{rem}
	Assume that the linear version of the proposed algorithm~\ref{alg_1} with $h_l(\mb{x}_j^k)=\mb{x}_j^k$ and $h_l(\mb{y}^{k}_j)=\mb{y}^{k}_j$ gives certain optimality gap. 
	Then, recalling Eq.~\eqref{eq_hl}, one can see that $\underline{\mc{K}} (\mb{x}_j^k -\mb{x}_i^k \leq h_l(\mb{x}_j^k)-h_l(\mb{x}_i^k) \leq \overline{\mc{K}} (\mb{x}_j^k -\mb{x}_i^k)$ and $\underline{\mc{K}} (\mb{y}_j^k -\mb{y}_i^k \leq h_l(\mb{y}_j^k)-h_l(\mb{y}_i^k) \leq \overline{\mc{K}} (\mb{y}_j^k -\mb{y}_i^k)$. 	
	Therefore, the optimality gap is of the same order as in the linear case multiplied by a factor of $\overline{\mc{K}}$ and $\underline{\mc{K}}$. Thus, for the logarithmic quantization nonlinearity with $\overline{\mc{K}} = 1+\frac{\rho}{2}$ and $\underline{\mc{K}}= 1-\frac{\rho}{2}$, the optimality gap gets arbitrarily close to the linear case by choosing sufficiently small $\rho$. This is illustrated in the next section via numerical simulations. 
\end{rem} 

\begin{rem}
	Uniform quantization in contrast to logarithmic quantization is not sector-bound; as shown in Fig.~\ref{fig_quant_desmos}, the bounds $x \pm \frac{\rho}{2}$ do not pass through the origin. Therefore, the results on convergence and optimality do not hold for uniformly quantized data. In fact, some works in the literature show that uniform quantization generally results in a larger optimality gap, see an example for quantized resource allocation in \cite{cpu}. In the next section, we clearly show by numerical simulation that log-scale quantization as compared to uniform quantization results in a smaller optimality gap. 
\end{rem}
	
\section{Numerical Experiments} \label{sec_exp}
	In this section, we perform simulations on both real and academic setups to verify our results.
	\subsection{Academic Example}
	We consider a network of $n=16$ agents to minimize the following cost function in a distributed and collaborative way:
	\begin{align}\label{eq_fij_sim}
		f_{i,j}(x_i) = 4 x_i^2 +3\sin^2(x_i)+a_{i,j} \cos(x_i) + b_{i,j}x_i,
	\end{align}
	with random parameters $a_{i,j}$ and $b_{i,j}$ in the range $[-10,10]$ such that $\sum_{i=1}^n \sum_{j=1}^m a_{i,j} = 0$ and $\sum_{i=1}^n \sum_{j=1}^m b_{i,j}=0$ and $a_{i,j},b_{i,j} \neq 0$. Note that $f_{i,j}(x_i)$ is not necessarily convex, i.e., $\nabla^2 f_{i}(\mb{x})$ might be negative at some points, while the global cost summing all $f_{i,j}(x_i)$ is strongly convex. We apply Algorithm~\ref{alg_1} over random ad-hoc networks and structured exponential networks for both uniform and log-scale quantization. First, we compare the convergence rate and optimality gap over a time-varying (switching) exponential network (as a structured network) versus a randomly changing Erdos-Renyi (ER) network subject to log-scale quantization with $\rho=\frac{1}{128}$. For the exponential network, we change the link weights every $100$ iterations and for the ER random network both the link weights and network structure change every $100$ iterations. The result is shown in Fig.~\ref{fig_acad_exprand}. The simulation results clearly verify that the solution converges over switching networks and also, as stated in Section~\ref{sec_effect}, the structured network gives a smaller optimality gap and faster convergence.
	\begin{figure} %[b]
		\centering
		\includegraphics[width=2.4in]{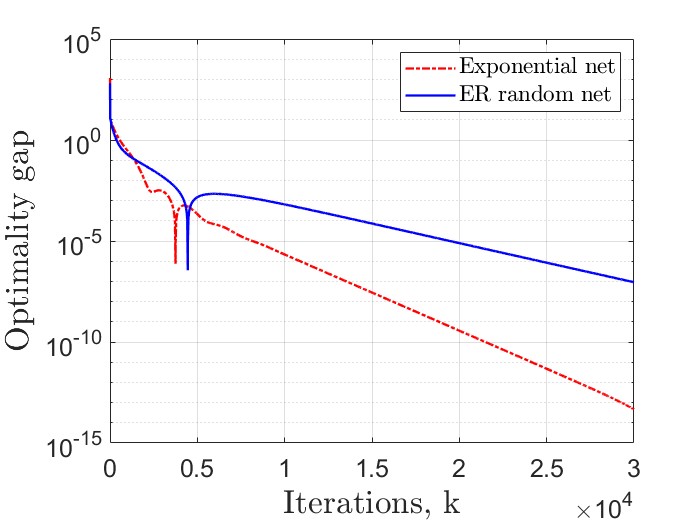}	
		\caption{Logarithmically quantized distributed optimization over an exponential network versus an ER random network of the same size.
		} \label{fig_acad_exprand}
	\end{figure}
	
	Next, we compare distributed optimization over exponential graphs under time-varying weights subject to uniform quantization versus log-scale quantization in Fig.~\ref{fig_acad_quant} for different quantization levels $\rho$. Evidently from the figure, log-scale quantization results in a smaller optimality gap as compared to uniform quantization which verifies the statements in Section~\ref{sec_effect}. This shows the advantage of using logarithmic quantization over uniform quantization. Also, one can see that the log-scale quantization level has negligible effect on the convergence rate and optimality gap. Note that, from Remark~\ref{rem_rate}, the convergence rate is linear in the log-scale. Therefore, one can approximate the convergence rate of the log-quantized case by the linear (non-quantized) solution. Also, note that considering the computational limitation of MATLAB the existing optimality gap of the log-quantized dynamics in the simulation is acceptable.
	\begin{figure} %[b]
		\centering
		\includegraphics[width=2.4in]{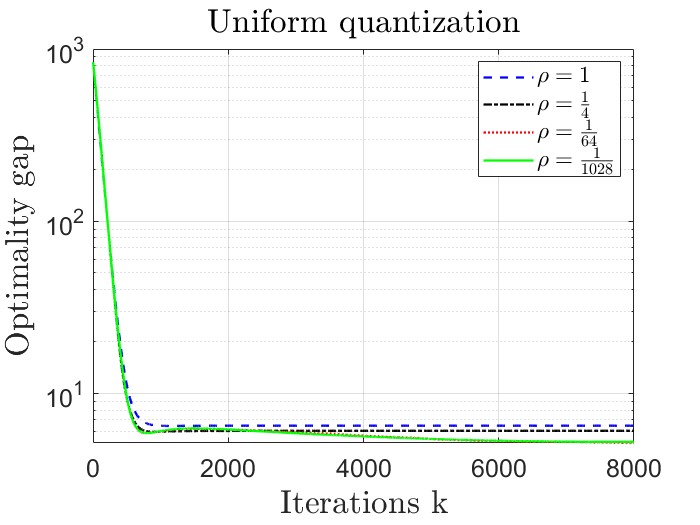}	\includegraphics[width=2.4in]{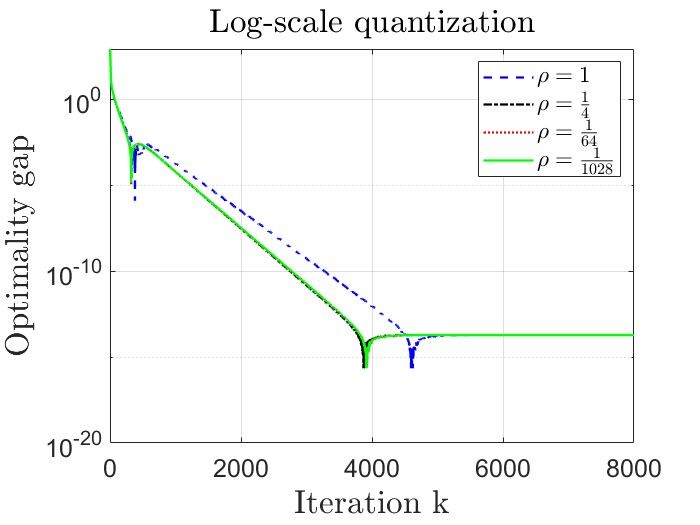}
		\caption{Optimality gap of uniformly quantized versus logarithmically quantized distributed learning over exponential network.
		} \label{fig_acad_quant}
	\end{figure}
	
	Next, we compare our proposed log-quantized algorithm with some existing literature: \textbf{TV-AB} \cite{saadatniaki2020decentralized}, finite-time \cite{yao2018distributed}, and fixed-time \cite{ning2017distributed} algorithms. It should be noted that none of these algorithms are quantized and these algorithms are designed to solve strongly-convex cost functions. Therefore, we need to decrease the non-convexity parameter of the cost model \eqref{eq_fij_sim} and set $a_{i,j} \in [-1,1]$ to make it strongly-convex for this simulation. For our log-quantized solution we set $\rho = \frac{1}{4}$. The comparison is given in Fig.~\ref{fig_acad_comp}. It is clear from the figure that our solution, although quantized, has the same performance as the non-quantized AB-type solution and shows better performance than fixed-time and finite-time solutions.  
	\begin{figure} %[b]
	\centering
	\includegraphics[width=2.4in]{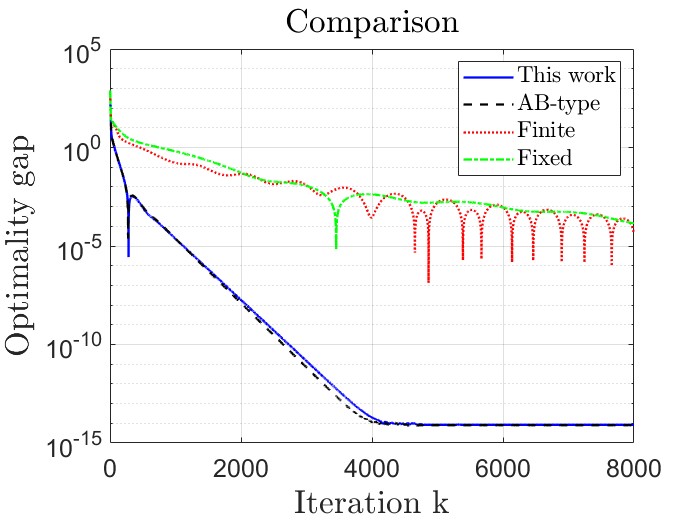}	
	\caption{The comparison of the proposed log-quantized algorithm with the existing non-quantized algorithms for a strongly-convex cost model.
	} \label{fig_acad_comp}
\end{figure}	

	\subsection{Real Data-Set Example}
	In this section, we analyze some real data from \cite{9827792}. We randomly select some images from the MNIST data set and classify them using logistic regression with a convex regularizer. The set of
	$N = 12000$ labelled images are used for classification, distributed
	among the $n=16$ agents/nodes. The overall cost function is defined as 
	\begin{align}
		\min_{\mb{b},c} &
		F(\mb{b},c) = \frac{1}{n}\sum_{i=1}^{n} f_i
	\end{align}  
	where each node $i$ has access to a batch of $m_i=750$ sample data and locally minimizes the following training cost:
	\begin{align}\label{eq_fij_regression}
		f_i(\mb{x}) = \frac{1}{m_i}\sum_{j=1}^{m_i} \ln(1+\exp(-(\mb{b}^\top x_{i,j}+c)y_{i,j}))+\frac{\lambda}{2}\|\mb{b}\|_2^2.
	\end{align}
	with $\mb{b},c$ as the parameters of the separating hyperplane.
	We run and compare the distributed training over both geometric  (as an ad-hoc network) and exponential graphs (as a structured network). The optimality gap (also known as the optimization residual) in the figures is defined as $F(\overline{\mb{x}}^k)-F(\mb{x}^*)$ with $\overline{\mb{x}}^k = \frac{1}{n} \sum_{i=1}^{n}\mb{x}^k_i$. 
	
	First, we compare some existing algorithms in the literature with the proposed Algorithm~\ref{alg_1}. The following algorithms are used for comparison: \textbf{GP} \cite{nedic2014distributed}, \textbf{SGP} \cite{spiridonoff2020robust,nedic2016stochastic}, \textbf{S-ADDOPT} \cite{qureshi2020s}, \textbf{ADDOPT} \cite{xi2017add}, and \textbf{PushSAGA} \cite{qureshi2021push}. 
	Fig.~\ref{fig_compare} presents the comparison results.      
	\begin{figure} %[b]
		\centering
		\includegraphics[width=2.4in]{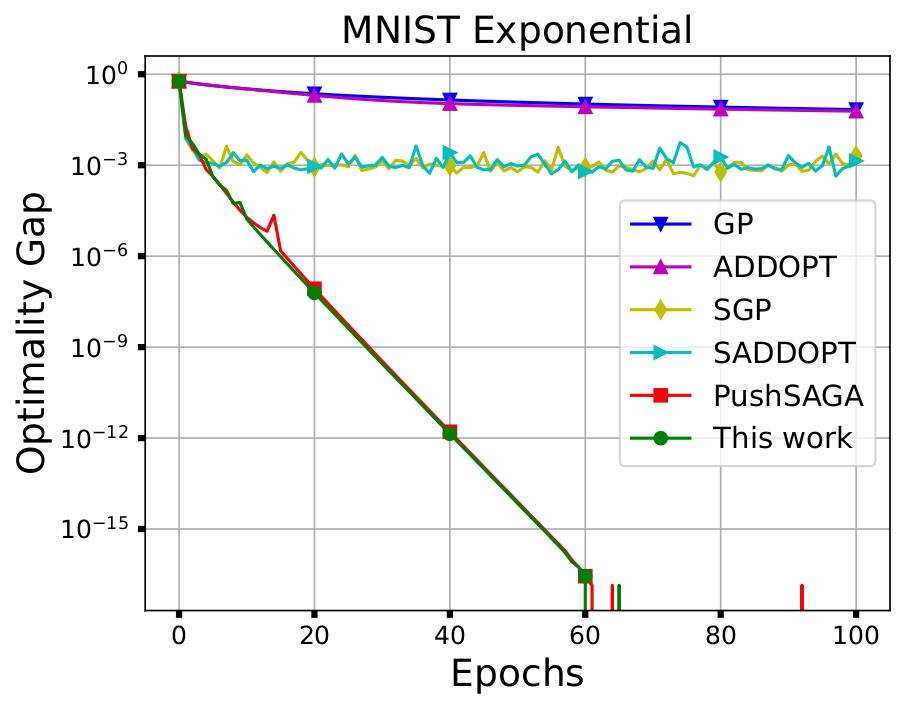}
		\includegraphics[width=2.4in]{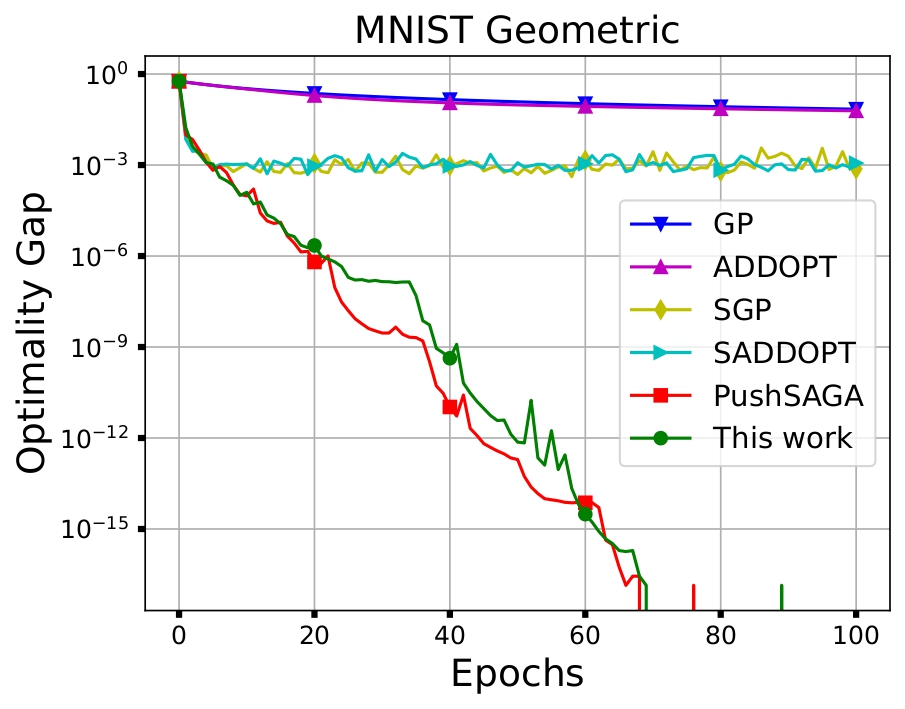}
		\caption{Comparison between the performance of the proposed algorithm (in the absence of quantization) with some existing algorithms in the literature.
		} \label{fig_compare}
	\end{figure}
	
	Next, we compare training over exponential networks subject to uniform and log-scale quantization for different data quantization levels, see Fig.~\ref{fig_rho_exp} and Fig.~\ref{fig_rho_exp2}. As it can be observed, the uniform quantization results in a larger optimality gap in the training outcome. 
	\begin{figure} %[b]
		\centering
		\includegraphics[width=2.4in]{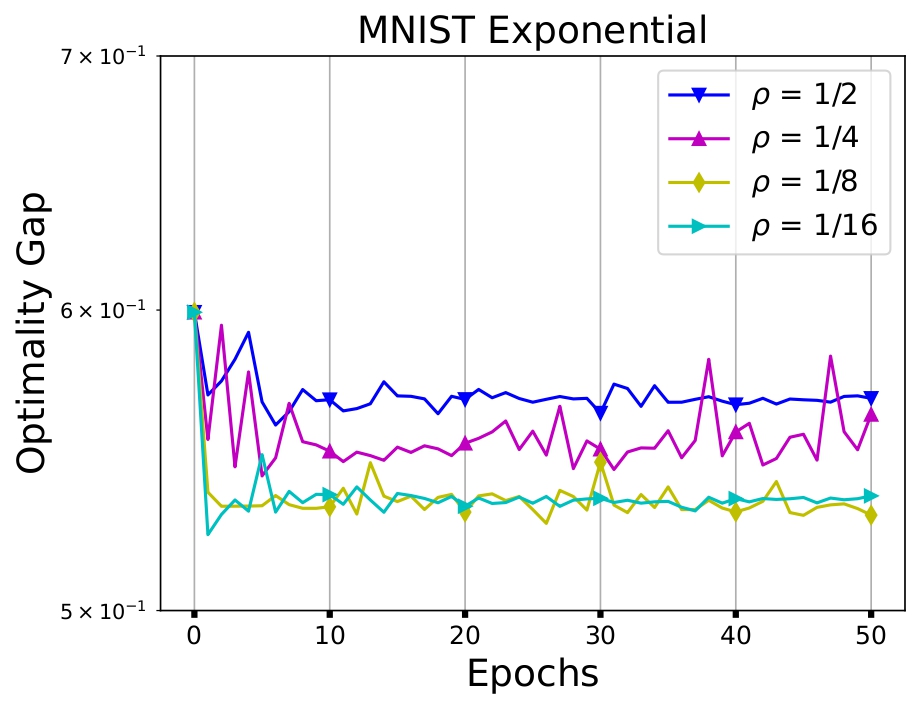}
		\caption{Optimality gap of distributed learning of MNIST dataset over exponential graphs subject to uniform quantization with different quantization levels.
		} \label{fig_rho_exp}
	\end{figure}
	\begin{figure} %[b]
	\centering
	\includegraphics[width=2.4in]{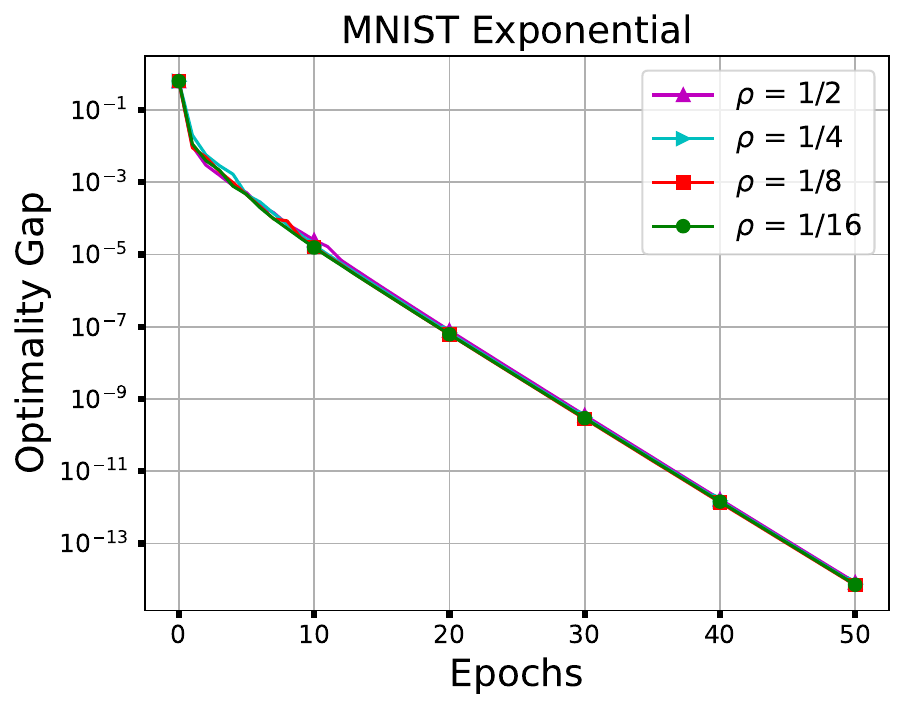}
	\caption{Optimality gap of distributed learning of MNIST dataset over exponential graphs subject to log-scale quantization with different quantization levels.
	} \label{fig_rho_exp2}
    \end{figure}
	The same comparison is performed over the geometric network and the result is shown in Fig.~\ref{fig_rho_geo} and Fig.~\ref{fig_rho_geo2}. As expected, uniform quantization results in a larger optimality gap. Also, from Fig.~\ref{fig_rho_exp2} and~\ref{fig_rho_geo2} one can see that the convergence rates for different log-scale quantization rates do not change that much. 
	It is worth noting from the figures that, the optimality gap is generally larger over geometric networks as a non-structured topology. However, the log-scale quantization adds more computational complexity to the algorithm as compared to the uniform quantization. This is because log-quantization is more precise around the origin and requires more bits for transmitting those values.
	\begin{figure} %[b]
		\centering
		\includegraphics[width=2.4in]{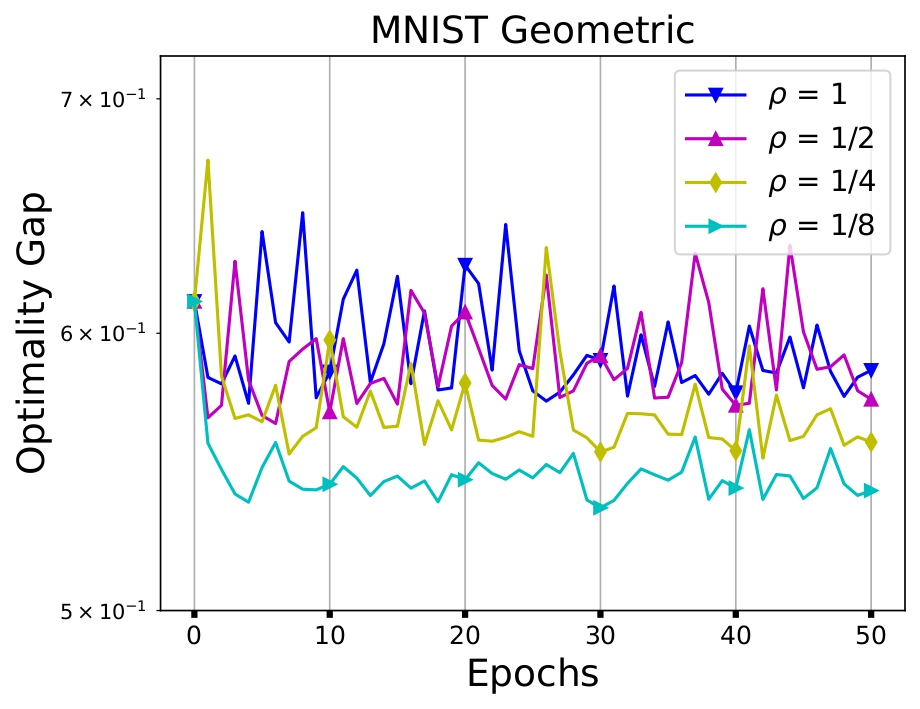}
		\caption{Optimality gap of distributed learning of MNIST dataset over geometric graphs subject to uniform quantization with different quantization levels.
		} \label{fig_rho_geo}
	\end{figure}
	\begin{figure} %[b]
	\centering
	\includegraphics[width=2.4in]{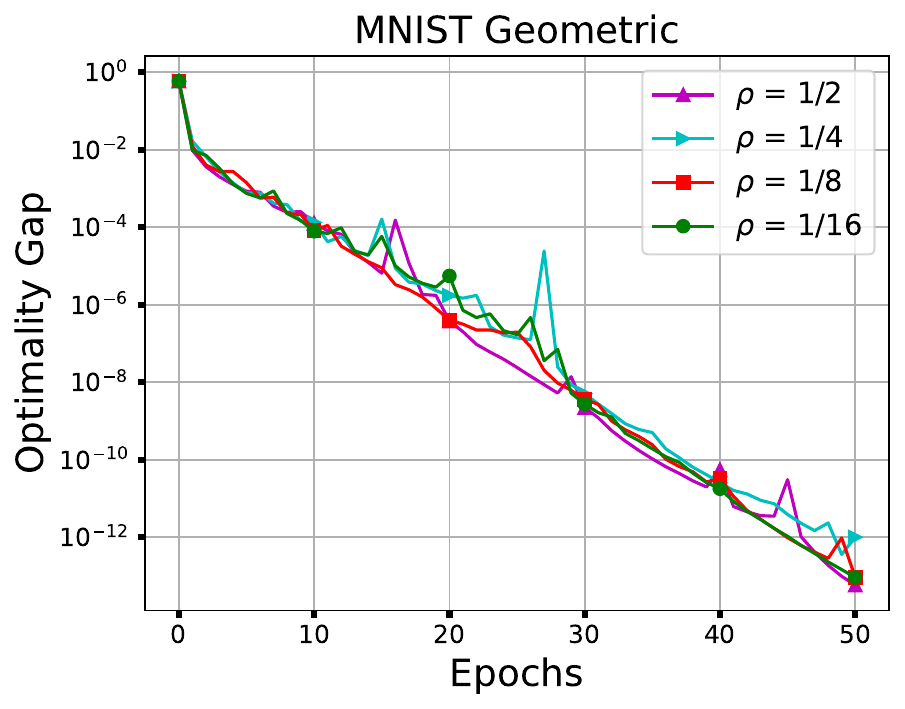}
	\caption{Optimality gap of distributed learning of MNIST dataset over geometric graphs subject to log-scale quantization with different quantization levels.
	} \label{fig_rho_geo2}
\end{figure}	
	\section{Conclusions}\label{sec_con}
	\subsection{Concluding Remarks}
	This work proposes a distributed optimization setup for learning over networks. We particularly address nonlinear channels to take into account, for example, possible data quantization. We compare log-scale quantization (as a sector-bound nonlinearity) with uniform quantization (as a non-sector-bound nonlinearity) over both academic and real data.  Our results show that the log-scale quantization leads to a smaller optimality gap for distributed training. We further compare the quantized setups over exponential and random (geometric and ER) networks, and, as expected, the exponential graph as a structured network results in a smaller optimality gap as compared to geometric and ER graphs representing ad-hoc multi-agent networks.
	
	\subsection{Future Directions}
	One direction of future research is to more relax Assumption~\ref{ass_wb} to uniform-connectivity of the network, where the union of the graphs over a finite interval is connected. More detailed theoretical analysis and comparison between log-scale versus uniform quantization for general distributed optimization methods and quantifying the optimality gap of the discretized version for large sampling periods are set as future research. 
	Other future research efforts will focus on considering the effect of quantization on different machine learning techniques, for example, distributed minimization of error-back-propagation in Neural-Networks and distributed reinforcement learning. Application to distributed optimal control is another direction of research interest.

	\bibliographystyle{IEEEbib}
	\bibliography{bibliography}
	
	\begin{IEEEbiography}[{\includegraphics[width=1.1in,clip,keepaspectratio]{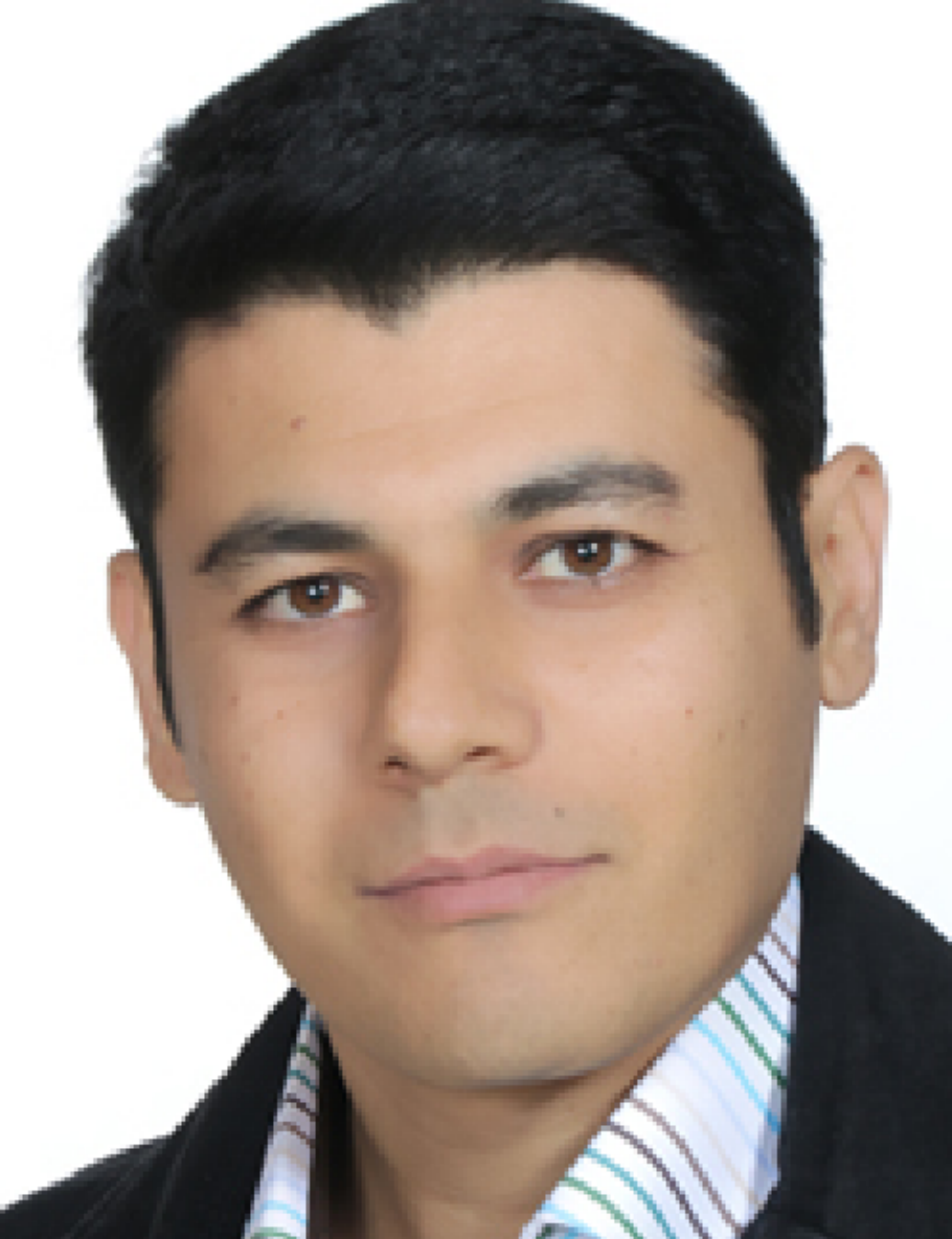}}]{Mohammadreza~Doostmohammadian}
		received his B.Sc. and M.Sc. in Mechanical Engineering from Sharif University of Technology, Iran, respectively in 2007 and 2010, where he worked on applications of control systems and robotics. He received his PhD in Electrical and Computer Engineering from Tufts University, MA, USA in 2015. During his PhD in Signal Processing and Robotic Networks (SPARTN) lab, he worked on control and signal processing over networks with applications in social networks. From 2015 to 2017 he was a postdoc researcher at ICT Innovation Center for Advanced Information and Communication Technology (AICT), School of Computer Engineering, Sharif University of Technology, with research on network epidemic, distributed algorithms, and complexity analysis of distributed estimation methods. He was a researcher at Iran Telecommunication Research Center (ITRC), Tehran, Iran in 2017 working on distributed control algorithms and estimation over IoT. Since 2017 he has been an Assistant Professor with the Mechatronics Department at Semnan University,  Iran. From 2021 to 2022, he was a visiting researcher at the School of Electrical Engineering and Automation, Aalto University, Espoo, Finland,  working on constrained and unconstrained distributed optimization techniques and their applications. His general research interests include distributed estimation, control, learning, and optimization over networks. He was the chair of the robotics and control session at the ISME-2018 conference, the session chair at the 1st Artificial Intelligent Systems Conference of Iran, 2022, and the Associate Editor of IEEE/IFAC CoDIT2024 conference.
	\end{IEEEbiography}

	\begin{IEEEbiography}[{\includegraphics[width=1.1in,clip,keepaspectratio]{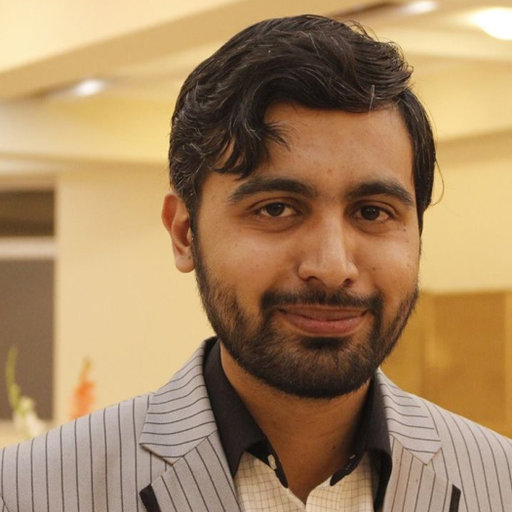}}]{Muhammad~I.~Qureshi}
	 received the B.S. degree in electrical engineering from the University of Engineering and Technology Lahore, Pakistan, in 2017, and the Ph.D. degree in electrical and computer engineering from Tufts University, USA, in 2024. He was working as a Research Scientist at Tufts University under the supervision of Prof. Usman A. Khan. Currently, he is a researcher at Intel corporation.
\end{IEEEbiography}

	\begin{IEEEbiography}[{\includegraphics[width=1.1in,clip,keepaspectratio]{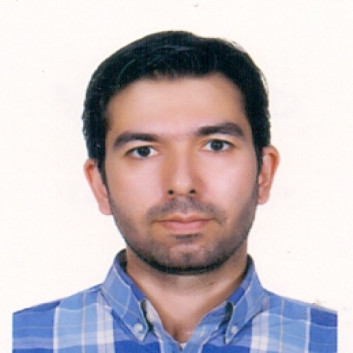}}]{Mohammad~Hossein~Khalesi}
		received his B.Sc. degree in Mechanical Engineering and also in Aerospace Engineering, M.Sc. degree in Mechanical Engineering, and Ph.D. in Mechanical Engineering all from Sharif University of Technology, Tehran, Iran, respectively in 2011, 2013, and 2019. Currently, he is an assistant professor of Mechanical Engineering at Semnan University, Semnan, Iran.
\end{IEEEbiography}

\begin{IEEEbiography}[{\includegraphics[width=1.1in]{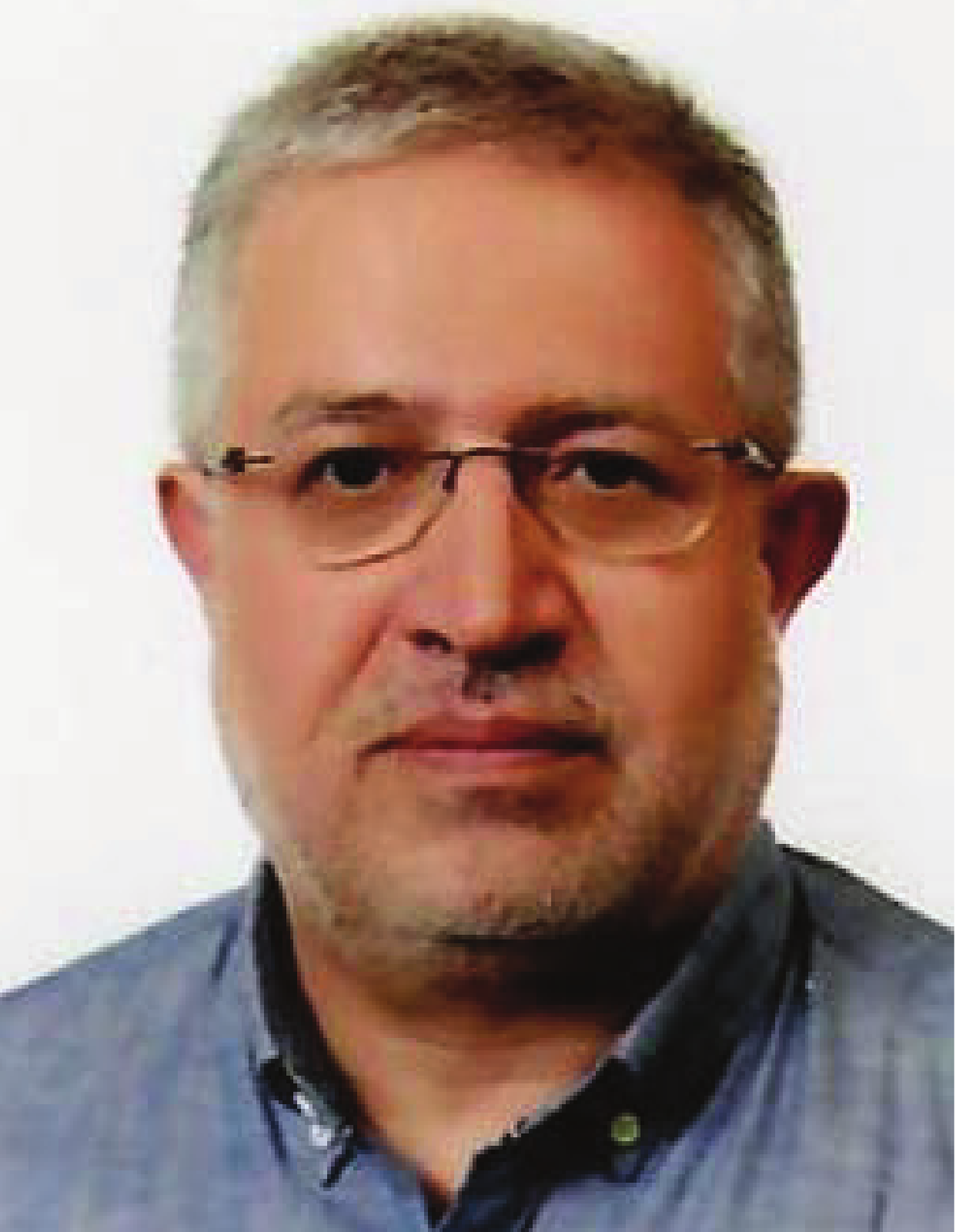}}]{Hamid R. Rabiee}
	received his BS and MS degrees (with Great Distinction) in Electrical Engineering from CSULB, Long Beach, CA (1987, 1989), his EEE degree in Electrical and Computer Engineering from USC, Los Angeles, CA (1993), and his Ph.D. in Electrical and Computer Engineering from Purdue University, West Lafayette, IN, in 1996. From 1993 to 1996 he was a Member of Technical Staff at AT\&T Bell Laboratories. From 1996 to 1999 he worked as a Senior Software Engineer at Intel Corporation. He was also with PSU, OGI and OSU universities as an adjunct professor of Electrical and Computer Engineering from 1996-2000. Since September 2000, he has joined Sharif University of Technology, Tehran, Iran. He was also a visiting professor at the Imperial College of London for the 2017-2018 academic year. He is the founder of Sharif University Advanced Information and Communication Technology Research Institute (AICT), ICT Innovation Center, Advanced Technologies Incubator (SATI), Digital Media Laboratory (DML), Mobile Value Added Services Laboratory (VASL), Bioinformatics and Computational Biology Laboratory (BCB) and Cognitive Neuroengineering Research Center. He has also been the founder of many successful High-Tech start-up companies in the field of ICT as an entrepreneur. He is currently a Professor of Computer Engineering at Sharif University of Technology, and Director of AICT, DML, and VASL. He has been the initiator and director of many national and international level projects in the context of Iran National ICT Development Plan and UNDP International Open Source Network (IOSN). He has received numerous awards and honors for his Industrial, scientific and academic contributions. He has acted as chairman in a number of national and international conferences, and holds three patents. He is also a Member of IFIP Working Group 10.3 on Concurrent Systems, and a Senior Member of IEEE. His research interests include statistical machine learning, Bayesian statistics, data analytics and complex networks with applications in social networks, multimedia systems, cloud and IoT privacy, bioinformatics, and brain networks.
\end{IEEEbiography}

	\begin{IEEEbiography}[{\includegraphics[width=1.1in]{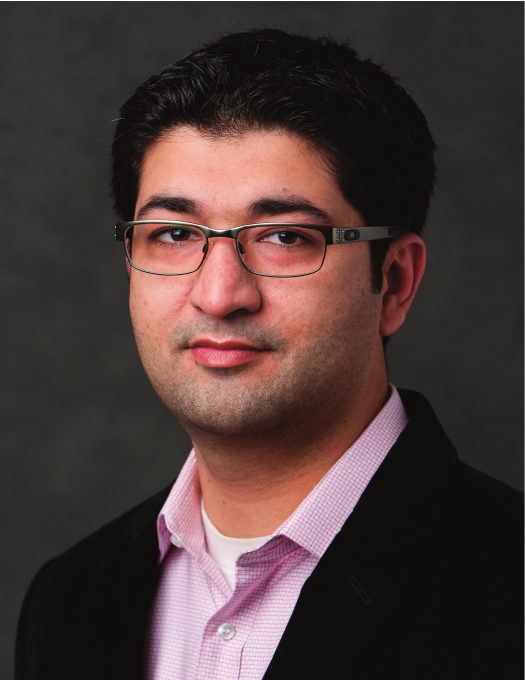}}]{Usman A. Khan} is a Professor of ECE and Computer Science at Tufts University. He is also an Amazon Scholar with Amazon Robotics. His research interests include artificial intelligence, machine learning, robotics, optimization and control, and stochastic dynamical systems. Recognition of his work includes the prestigious NSF Career award, several federally funded projects and NSF REU awards, an IEEE journal cover. He received the 2023 Tufts Inventor award, the 2022 EURASIP Best Paper Award for articles published in the \textit{EURASIP Journal on Advances in Signal Processing} over 2017-2021, and Best Student Paper awards at the \textit{IEEE International Conference on Networking, Sensing and Control} (2014), and at the \textit{IEEE Asilomar Conference} (2014 and 2016). 
 Dr. Khan received his B.S. degree in 2002 from University of Engineering and Technology, Pakistan, M.S. degree in 2004 from University of Wisconsin, and Ph.D. degree in 2009 from CMU, all in ECE. He was a postdoc in the GRASP lab at the U-Penn and also has held a Visiting position at KTH, Sweden. Dr. Khan is an \textit{IEEE Senior Member} and was an elected full member of the \textit{Sensor Array and Multichannel TC} with the \textit{IEEE Signal Processing Society} from 2019-2022, where he was an Associate member from 2010 to 2019. He was an elected full member of the \textit{IEEE Big Data Special Interest Group} from 2017 to 2019, and has served on the \textit{IEEE Young Professionals Committee} and on \textit{the IEEE Technical Activities Board}. He has served as an Associate Editor on several top-tier IEEE publications: \textit{IEEE Transactions on Smart Grid} (2014-2017); \textit{IEEE Control System Letters} (2018-2020), \textit{IEEE Transactions on Signal and Information Processing over Networks} (2019-current), where he received the 2023 Outstanding Editorial Member award; \textit{IEEE Open Journal of Signal Processing} (2019-current); and \textit{IEEE Transactions on Signal Processing} (2021-2023), where currently he is a Senior Area Editor. He served as the Chief Editor for the \textit{Proceedings of the IEEE} special issue on \textit{Optimization for Data-driven Learning and Control} (Nov. 2020), and as a Guest Associate Editor for the \textit{IEEE Control System Letters} special issue on \textit{Learning and Control} (Nov. 2020). He served as the Technical Area Chair for the Networks track in \textit{2020 IEEE Asilomar Conference on Signals Systems and Computers} and for the \textit{Signal Processing for Self Aware and Social Autonomous Systems} track at the \textit{1st IEEE International Conference on Autonomous Systems} (Aug. 2021).
		 
	\end{IEEEbiography}
\end{document}